%% file: report.tex
\documentclass[a4]{article}

\usepackage{amsmath}
\usepackage{amssymb}
\usepackage[utf8]{inputenc}
\usepackage{subfig}
\usepackage{microtype}
\usepackage{caption}
\usepackage{algpseudocode}
\usepackage{algorithm}
\usepackage{listings}
\usepackage{mathtools}
\usepackage{mymacros}
\usepackage{marginnote}
\usepackage{multicol}

\usepackage{soul}
\usepackage{color}

\newcommand{\added}[1]{#1}
\newcommand{\deleted}[1]{}
\newcommand{\substituted}[2]{#2}
\newcommand{\commentaire}[1]{}

\usepackage{tikz}
\usetikzlibrary{arrows}
\usetikzlibrary{positioning}
\usetikzlibrary{patterns}
\usetikzlibrary{shapes}

\newtheorem{mydef}{Definition}
\newtheorem{property}{Property}
\newtheorem{theorem}{Theorem}
\newtheorem{lemma}{Lemma}
\newtheorem{example}{Example}

\newenvironment{proof}{\textbf{Proof.}}{ $\square$ }

\title{Symbolic Computation of the Worst-case Execution Time of a Program\footnote{This is a preliminary version of the article ``Simbolic WCET Computation'' from the same authors, accepted for publication at ACM Transactions on Embedded Computing Systems}\\[.5cm]
  \large Technical Report } \author{Cl\'ement Ballabriga, Julien
  Forget, Giuseppe Lipari\\ Univ.Lille, CNRS, Centrale Lille UMR 9189
  - CRIStAL}

\date{}

\begin{document}

\maketitle

\begin{abstract}
  Parametric Worst-case execution time (WCET) analysis of a sequential
  program produces a formula that represents the worst-case execution
  time of the program, where parameters of the formula are
  user-defined parameters of the program (as loop bounds, values of
  inputs or internal variables, etc).

  In this paper we propose a novel methodology to compute the
  parametric WCET of a program. Unlike other algorithms in the
  literature, our method is not based on Integer Linear Programming
  (ILP). Instead, we follow an approach based on the notion of
  symbolic computation of WCET formulae. After explaining our
  methodology and proving its correctness, we present a set of
  experiments to compare our method against the state of the art. We
  show that our approach dominates other parametric analyses, and
  produces results that are very close to those produced by
  non-parametric ILP-based approaches, while keeping very good
  computing time.
\end{abstract}

\section{Introduction}
\input{introduction}

\section{Related works}
\input{related}

\input{cfg}

\input{tree}

\input{abstract-wcet}

\input{parametric}

\section{Experiments}
\input{experiments}

\section{Conclusion}
\input{conclusion}

\bibliographystyle{plain}
\bibliography{biblio}

\appendix

\input{proofs}

\end{document}

%% file: introduction.tex
A real-time system is usually represented as a set of tasks. Tasks are
subject to timing constraints: typically, the execution of every
instance of a periodic real-time task must be completed before its
\emph{deadline}. In order to guarantee the respect of timing
constraints, first a \emph{worst-case execution time} (WCET) analysis is
performed off-line, which calculates an upper bound to the execution
time of each task. Then, this information is used to perform a
schedulability analysis and guarantee that every task will meet its
deadline.

In this paper, we focus on WCET analysis. In WCET analysis, first the
task code is analysed to model its set of possible execution
paths. Then, the impact of the hardware architecture is taken into
account: local effects (timing of basic blocks of code) and global
effects (impact of processor pipeline, caches, and in general
interactions between basic blocks). Finally, an upper bound to the
execution time is computed by calculating the worst-case path, taking
into account all effects. A popular technique for doing this, called
Implicit Path Enumeration Technique (IPET), is to encode the problem as
an Integer Linear Programming (ILP) problem that is then solved with
standard techniques~\cite{ipet}.

With traditional WCET analysis, if any of the program parameters is
changed, it is necessary to re-run the analysis. Also, it is difficult
to analyze the impact of different parameter values on the final WCET
estimate. For example, the developer may be want to know the impact of
the number of iterations of a certain loop on the WCET, the impact of
the cache size, etc. To answer these questions, it would be necessary to
run the analysis several times with different parameter values, which
could be a very time consuming process.

An alternative approach is to calculate directly a \emph{parametric
  WCET} formula instead of a constant value. If the parameter changes,
it is possible to recompute the WCET by simply substituting the
parameter value into the formula. Thus, it is possible to quickly
explore the parameters space, which may be very useful in guiding
developers at design time. Similarly, parametric WCET simplifies the
analysis process when third-party software is involved, since the
developer can provide a parametric WCET along with the component, that
can be adapted to the target system.

In addition, if the obtained formula is simple enough, it can be used to
efficiently implement an \emph{adaptive} real-time system. Indeed, many
system parameters are only known at run-time: loop bounds that depend on
input values, software and hardware state changes, operating system
interference, etc. With traditional WCET analysis, adaptive features
would rely on a pre-computed WCET table containing different WCET values
for different parameter values. Instead, with parametric WCET analysis,
we can compute off-line a WCET formula that depends on these parameters
and instantiate this formula on-line, at which point parameter values
become known. As a result, with low overhead, we obtain a tighter
estimate of the task's WCET and take better scheduling decisions. This
can for instance benefit energy-aware scheduling techniques based on
Dynamic Voltage and Frequency Scaling (DVFS)~\cite{parascale}.

Finally, large execution time values may happen only very rarely, for
instance for unlikely combinations of input data. By using parametric
WCET analysis, it is possible to design the system according to an
upper bound that is safe for the vast majority of executions of the
system, and then evaluate a parametric WCET formula at run-time to
trigger an alternate, less time-consuming computation when the formula
returns a value exceeding the safe bound (and thus remain under the
safe bound).

\textbf{Contribution.} In this paper, we propose a novel approach to
parametric WCET analysis based on \emph{symbolic computation} that
greatly improves upon the state of the art on parametric WCET. Unlike
the majority of existing WCET analysis algorithms, our methodology is
not based on ILP: instead, we follow an approach based on symbolic
computation of WCET formulae.

We start from a representation of the program as a \emph{Control-Flow
  Graph} where nodes of the graph are basic blocks of code (the notion
of CFG is recalled in Section~\ref{sub:ccfg}). We transform the CFG into
a \emph{Control-Flow Tree} (CFT) (Section \ref{sub:expressiontree}),
because a tree is more amenable to be transformed into arithmetic
(symbolic) formulae.  To represent global effects, CFT nodes are
annotated with \emph{context-sensitive annotations} (Section
\ref{sub:context-annotation}): these annotations encode restrictions on
the number of iterations of basic blocks when executed inside
loops. They may be considered as the equivalent of ILP constraints in
the IPET method~\cite{ipet}.  We then move to the core method for
generating a WCET formula.  We first
introduce the notion of \emph{Abstract WCET} (Section
\ref{sec:abstract}) and how to compute it starting from an annotated CFT
in the absence of parameters. Later, we introduce WCET parameters
(Section \ref{sec:parametric}) and we enunciate the rules for symbolic
computation and simplification of \emph{Abstract WCET formulae}.
Finally, in Section \ref{sec:eval} we present experimental data that
compare our approach with the state of the art algorithms. We show that
our algorithm produces results that are very close to those of
non-parametric ILP-based approaches, while keeping very good computing
time. We also show that simplified WCET formulae are very small, which
implies low memory and execution time overhead in case of on-line
formula evaluation. Finally, we show that our approach dominates other
parametric WCET analyses. This paper focuses on the generic framework
for symbolic WCET evaluation and only briefly \substituted{describes}{outlines} some applications
\added{in Section~\ref{sec:symb-ex}}. More complex applications
(e.g. data-cache analysis) are out of the scope of this paper and are
subject to future work.

\begin{figure}[tbp]
\centering
  \begin{tikzpicture}[scale=0.7, transform shape,node distance=0.6cm]
    \tikzset{vertex/.style = {shape=ellipse,draw,minimum size=1.5em}}
    \tikzset{vertex2/.style = {shape=rectangle,draw,minimum size=1.5em}}
    \tikzset{edge/.style = {->,> = latex'}}
	\node[vertex2, align=center] (cfg) {CFG};
	\node[vertex, align=center] (tb) [right= of cfg] {Tree\\Builder};
	\node[vertex2, align=center] (cft) [right= of tb] {CFT};
	\node[vertex, align=center] (comp) [right= of cft] {Abs. WCET\\computation};
	\node[vertex2, align=center] (awcet) [right= of comp] {Abs. WCET};
	\node[vertex, align=center] (apply) [right= of awcet] {Abs. WCET\\instantiate};
	\node[vertex2, align=center] (wcet) [right= of apply] {WCET};
	\node[vertex2, align=center] (param) [above= of apply] {Parameter\\values};
	\node[vertex, align=center] (ext) [above= of cft] {Extra-CFG\\analyses};

	\draw[edge] (cfg) to (tb);  
	\draw[edge] (tb) to (cft);  
	\draw[edge] (cft) to (comp);  
	\draw[edge] (comp) to (awcet);  
	\draw[edge] (awcet) to (apply);  
	\draw[edge] (apply) to (wcet);  
	\draw[edge] (param) to (apply);  
	\draw[edge] (ext) to (cft);  
  \end{tikzpicture}
  \caption{Symbolic WCET workflow}
  \label{pic:workflow}
\end{figure}
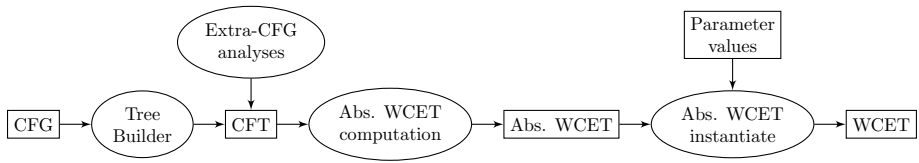


%% file: related.tex
Various existing works suggest using symbolic methods in WCET
analysis. However, their goal differs from ours. For example,
\cite{symbolic0,symbolic1,symbolic2} use symbolic execution as a
method to reduce the duration of the WCET
analysis. In~\cite{symbolic3}, the authors use symbolic states to
model the effect of pipelines on the WCET. The objective of these
papers is not to produce a parametric WCET formula.

In~\cite{blackbox}, a technique is presented to perform a partial,
composable WCET analysis. This work addresses mostly the software and
hardware modeling that occurs before the WCET computation
proper. Results are presented for the instruction cache and branch
prediction analysis, and loop bounds estimation. However, no solution
is provided to perform the ILP computation parametrically.

Feautrier~\cite{feautrier} presented a method for parametric ILP
computation. The ILP solver presented in~\cite{feautrier} (called
\emph{PIPLib}), takes a parametrized ILP system as input, and produces a
\emph{quast} (quasi-affine selection tree). Once computed, this tree can
be evaluated for any valid parameter values, without having to re-run
the solver. However, this approach is computationally very
expensive. 
Experiments~\cite{mpa} have shown that PIPLib does not scale well
when applied in the context of IPET. The MPA (Minimum Propagation
Algorithm)~\cite{mpa} attempts to address these
shortcomings. MPA takes as
input the results of the software and hardware modeling analysis, and
produces directly a parametric WCET formula. Compared with MPA, our
method is significantly tighter because it takes into account various
context-sensitive software and hardware timing effects.

In the past, many tree-based WCET computation methods have been
presented~\cite{ets}. In \cite{scope}, the authors suggest a method to
compute parametric WCETs using a tree-based approach. Our approach is
also based on trees, but unlike~\cite{scope} it can work directly on the
binary code. Furthermore, our method can model timing effects in a more
generic and accurate way thanks to context annotations (Section
\ref{sub:context-annotation}).

ParaScale~\cite{parascale} is an approach to exploit variability in
execution time to save energy. By statically analyzing the tasks, a
parametric WCET formula is given for loops in terms of the loop
iteration count. At run-time, before entering a loop, the formula is
evaluated and the system dynamically scales the voltage and frequency
of the processor. In comparison, the parameters in our method are not
limited to loop bounds.

Finally, note that our method provides an alternative to the
time-consuming ILP solving, thus our method is competitive even
compared to non-parametric WCET analysis based on ILP.


%% file: cfg.tex
\section{Control-Flow Graph}
\label{sub:ccfg}

In this section we recall the definition of Control-Flow Graphs (CFG),
the input model in our approach. The CFG is extracted from the binary
code of the task under analysis.

\begin{mydef}
  \label{def:cfg}
  A Control Flow Graph (CFG) is a directed graph $G=<\mathcal{B},
  \mathcal{E}>$. The set of vertices $\mathcal{B}$ corresponds to the
  set of basic blocks of the program represented by the CFG. A directed
  edge $(b_i, b_j)\in\mathcal{E}$ (where
  $\mathcal{E}\subseteq\mathcal{B}\times\mathcal{B}$), represents a
  valid succession of two basic blocks in the program execution. We
  denote by $\wcet{b}$ the worst-case execution time (WCET) of block
  $b$.
\end{mydef}

An \emph{entry node} is a node without incoming edges, and an \emph{exit
  node} is a node without outgoing edges. We assume, without loss of
generality, that a CFG has one single entry node and one single exit
node (otherwise, it is always possible to add fictive entry and exit
nodes with the corresponding edges). We also assume that each node is
reachable from the entry node, and that the exit node is reachable from
any node.

An \emph{execution path} is a sequence of nodes (basic blocks): $p.b$
denotes a path whose last node is $b$;
$p_1 \cat p_2$ denotes the path consisting of path $p_1$ followed by path
$p_2$. By abuse of notation, we also denote $b$ the path
consisting only of node $b$. $\epsilon$ denotes the empty path.

\begin{mydef}
  Let $G=<\mathcal{B}, \mathcal{E}>$ be a CFG. Let $p=b_1\ldots b_k$ an
  execution path. We say that $p$ is a \emph{valid path} of $G$ (or
  simply a path of $G$) iff:

  \[ \forall i\in \{1, \ldots, k\},
  b_i \in \mathcal{B} \\\wedge\forall i\in \{1, \ldots, k-1\}, (b_i,
  b_{i+1}) \in \mathcal{E}\]
 
  If $b_1$ is an entry node of $G$ and $b_n$ is an exit node of $G$, then $p$
  represents a complete execution of the program represented by $G$.  
\end{mydef}

\begin{mydef}
  Let $p=b_1\ldots b_k$ an execution path. We have: $\wcet{p}\equiv \sum_{i=1}^k \wcet{b_k}$
\end{mydef}

We introduce now a set of additional definitions concerning the CFG
topology that will allow us to manipulate the CFG in the following
sections.

\begin{mydef}
  Let $G=<\mathcal{B},\mathcal{E}>$. Let $b_i,b_j,h\in\mathcal{B}$ and
  let $h$ be a loop header (see definition below).
  \begin{itemize}
  \item We say that $b_i$ is a \emph{predecessor} of
     $b_j$, and denote $b_i\pred b_j$, iff $(b_i, b_j)\in\mathcal{E}$;
   \item We say that $b_i$ \emph{dominates} $b_j$, and denote $b_i\dom
     b_j$, iff all paths from the entry node to $b_j$ go through $b_i$;
   \item $b_k$ is the \emph{immediate dominator} of $b_i$ 
     iff $b_k\dom b_i$, $b_k \ne b_i$ and there exists no
	$b_{k'}$ such that $b_{k'} \ne b_i$, $b_{k'} \ne b_k$, $b_{k'}\dom b_i$, $b_k\dom b_{k'}$;
   \item $h$ is a \emph{loop header} if it has at least one
     predecessor $b_i$ such that $h \dom b_i$. We denote
     $l_h$ the loop associated to header $h$;
   \item An edge $(b_i,h)$ such that $h\dom b_i$ is called a
     \emph{back-edge} of $l_h$;
  \item An edge $(b_i,h)$ that is not a back-edge
    is called an \emph{entry-edge} of $l_h$.
  \item The \emph{body} of the loop of header $h$, denoted $\body{h}$,
    is the set of all nodes $b_i$ such that $b_i$ belongs to a path
    $P$, where $P$ starts with $h$, ends with a back-edge of $l_h$ and
	does not go through any entry-edges of $l_h$.
  \item An edge $(b_i,b_j)$ such that $b_i\in\body{h}$ and
    $b_j\not\in\body{h}$ is called an \emph{exit-edge} of $l_h$;
  \item An execution path of a loop $l_h$ is a path $p=h.b_1\ldots
	  b_n$, where there exists an exit edge $(b_n, b_x)$ of $l_h$.
    Note that $b_1$, $b_n$, $h$ may
    actually not be distinct. The number of iterations of $l_h$ in $p$
    corresponds to the number of back-edges in $p$. The maximum number
    of iterations of the loop $l_h$, denoted by $x_h$, is the maximum of
    the number of iterations of any execution path of
    $l_h$.
  \item Let $l_h$, $l_{h'}$ be two loops of $G$. We say that $l_h$
    \emph{contains} $l_{h'}$ and denote $l_{h'}\sqsubseteq l_h$ iff
    $h'\in\body{h}$;
  \item The loop $l_h$ \emph{immediately contains} $b_i$ iff
    $b_i\in\body{h}$ and there exists no loop $l_{h'}\neq l_h$ such
    that $b_i\in\body{h'}$ and $l_{h'}\sqsubseteq l_h$.
  \item The set of loops of graph $G$ is denoted $L_G$.
  \end{itemize}
\end{mydef}

We define two additional loops, that are not actually
part of the represented program:
\begin{itemize}
\item $\top$ is such that for all $l\in L_G$,
  $l\sqsubseteq \top$. In other words, $\top$ is a fictive loop
    whose body is the whole CFG ($\body{\top}=G$);
\item $\bot$ is such that for all $l\in L_G$,  $\bot\sqsubseteq
  l$. In other words, $\bot$ is a fictive empty loop ($\body{\bot}=\emptyset$).
\end{itemize}

\begin{property}
  $(L'_{G}=L_G\cup \{\top,\bot\},\sqsubseteq)$ is a lattice.
\end{property}
\begin{proof}
  Trivial due to the definition of $\top$ and $\bot$.
\end{proof}

In the following:
\begin{itemize}
  \item $\sqcup : L'_G \times L'_G \rightarrow L'_G$ denotes the least upper
    bound, i.e. $l_1 \sqcup l_2$ is the least element of
    $\{l \in L'_G | l_1 \sqsubseteq l \wedge l_2 \sqsubseteq l\}$.
  \item $\sqcap : L'_G \times L'_G \rightarrow L'_G$ denotes the greatest
    lower bound, i.e. $l_1 \sqcap l_2$ is the greatest element of
     $\{l \in L'_G | l \sqsubseteq l_1 \wedge l \sqsubseteq
    l_2\}$.
\end{itemize}

Figure \ref{fig:loops} shows of a simple CFG. Nodes
$b_1$ and $b_2$ are loop headers. Loop $l_{b_1}$ contains $b_1, b_2, b_3,
b_4$, $b_6$, but it immediately contains only $b_1$ and $b_3$. $(b_3,
b_1)$ is a back-edge and $(b_1, b_5)$ is an exit-edge for loop
$l_{b_1}$. Loop $l_{b_2}$ is contained within loop $l_{b_1}$. $b_1$ dominates
all the other nodes of the CFG. 
$b_1$ is the immediate dominator of $b_3$.

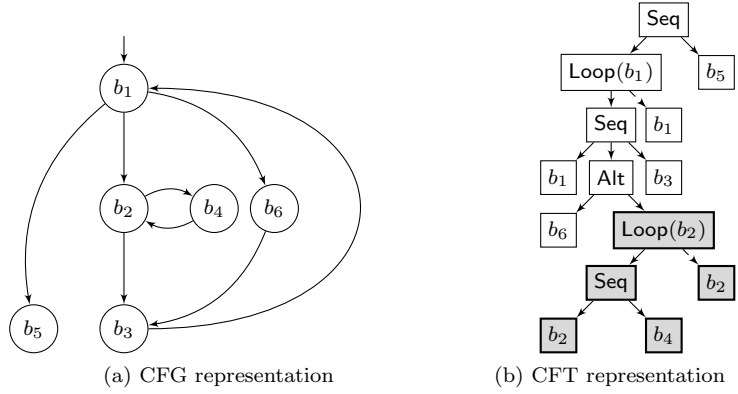
\begin{figure}[tbp]
\footnotesize
  \begin{center}
    \subfloat[CFG representation]{
  \begin{tikzpicture}[scale=0.8]
    \tikzset{vertex/.style = {shape=circle,draw,minimum size=1.5em}}
    \tikzset{edge/.style = {->,> = latex'}}
    \node (i) at (0,6) {};
    \node[vertex] (a) at  (0,5) {$b_1$};
	\node[vertex] (b) at  (0,3) {$b_2$};
    \node[vertex] (c) at  (0,1) {$b_3$};
    \node[vertex] (d) at  (1.5,3) {$b_4$};
    \node[vertex] (e) at  (-1.5,1) {$b_5$};
    \node[vertex] (f) at  (2.5,3) {$b_6$};
    \draw[edge] (i) to (a);
    \draw[edge] (a) to (b);
    \draw[edge] (b) to (c);
    \draw[edge, bend left] (b) to (d);
    \draw[edge, bend left] (d) to (b);
    \draw[edge, bend right=90,looseness=3] (c) to (a);
    \draw[edge, bend right] (a) to (e);
    \draw[edge, bend left] (a) to (f);
	\draw[edge, bend left] (f) to (c);

  \end{tikzpicture}
  \label{fig:loops}
}
\subfloat[CFT representation]{
\footnotesize
\begin{tikzpicture}[scale=0.7]
    \tikzset{vertex/.style = {shape=rectangle,draw,minimum size=1.5em}}
    \tikzset{vertex2/.style = {shape=rectangle,draw,thick,fill=gray!30,minimum size=1.5em}}
    \tikzset{edge/.style = {->,> = latex'}}
    \tikzset{edge_exit/.style = {->,> = latex', dashed}}
	\node[vertex] (m) at  (1, 3) {$\Seq$};
	\node[vertex] (n) at  (0, 2) {$\MyLoop(b_1)$};
        \node[vertex] (o) at  (1, 1) {$b_1$};
	\node[vertex] (z) at  (2, 2) {$b_5$};
        \node[vertex] (a) at  (-1,0) {$b_1$};
	\node[vertex] (b) at  (0,0) {$\Alt$};
	\node[vertex] (d) at  (1,0) {$b_3$};
        \node[vertex] (c) at  (0,1) {$\Seq$};
	\node[vertex] (f) at (-1,-1) {$b_6$};
        \node[vertex2] (g) at (1, -1) {$\MyLoop(b_2)$};
	\node[vertex2] (h) at (2, -2) {$b_2$};
	\node[vertex2] (i) at (0, -2) {$\Seq$};
	\node[vertex2] (j) at (-1, -3) {$b_2$};
	\node[vertex2] (k) at (1, -3) {$b_4$};
        \node (zz1) at (3,-2) {};
        \node (zz2) at (-3,-2) {};
	\draw[edge] (c) to (a);
	\draw[edge] (c) to (b);
	\draw[edge] (c) to (d);
	\draw[edge] (b) to (g);
	\draw[edge] (b) to (f);
        \draw[edge_exit] (g) to (h);
	\draw[edge] (g) to (i);
	\draw[edge] (i) to (j);
	\draw[edge] (i) to (k);
	\draw[edge] (m) to (n);
        \draw[edge_exit] (n) to (o);
	\draw[edge] (m) to (z);
	\draw[edge] (n) to (c);
  \end{tikzpicture}
  \label{fig:treefull}
}
\end{center}
\caption{A program with two nested loops.}
\end{figure}


%% file: tree.tex
\section{Control-Flow Tree}
\label{sub:expressiontree}

We propose to translate the CFG into a \emph{Control-Flow Tree}, which
also represents the possible execution paths of a program but, thanks to
its tree structure, is more prone to recursive WCET analysis than a
CFG. A Control-flow Tree is similar to Abstract Syntax Trees used in
programming languages compilation, except that it represents the
structure of binary code. As such, it will be quite natural to represent
the WCET of a CFT as an arithmetic expression (see Section~\ref{sec:parametric}).

\subsection{Definition}

The set of Control-flow Trees $\mathcal{T}$ is defined inductively as
follows:

\begin{mydef}
  Let $n,m\in\mathbb{N}^*$, $t_1$, $\ldots$, $t_n$,
  $\in\mathcal{T}^n$.  A control-flow tree $t\in\mathcal{T}$ is one
  of:
  \begin{itemize}
  \item $\Leaf(b)$, which represents the execution of basic block
    $b\in\mathcal{B}$; 
  \item $\Alt(t_1,\ldots,t_n)$, which represents an alternative
    between the execution of trees $t_1$, $\ldots$, $t_n$;
  \item $\MyLoop(h, t_1, n, t_2)$, which represents a loop with header $h$,
    that repeats the execution of tree $t_1$, with a maximum number
    of iterations $n$, and exits from the loop executing the
    tree $t_2$; 
  \item $\Seq(t_1,\ldots,t_n)$, which represents a sequential
    execution of trees $t_1$, $\ldots$, $t_n$.
  \end{itemize}
\end{mydef}

As an example, Figure~\ref{fig:treefull} shows the tree corresponding to
the CFG of Figure~\ref{fig:loops}. In the following sections, we will
use this example to describe the steps of the conversion
from CFG to CFT. Our definition of loops considers that we
repeat a sub-tree and then execute a different sub-tree when finishing
the loop. This enables to represent a wide variety of loops: $for$,
$while$, $do...while$, etc.

\subsection{From CFG to Control-flow Tree}

Algorithm~\ref{alg:dag} translates a loop of
the CFG into a \emph{Directed Acyclic Graph} (DAG) that represents the
loop body. Algorithm~\ref{alg:processing} is the recursive
procedure that generates the complete control-flow tree. It relies on
Algorithm~\ref{alg:dag} to process the CFG loops.

Our control-flow tree construction method works only for CFGs that
contain no irreducible loops (i.e. loops with multiples
entries). In the general case, it is possible to transform CFGs
with irreducible loops 
by using \emph{node splitting}~\cite{nodesplitting}
algorithms. In~\cite{decompil} the authors show that it is possible
to detect the set of irreducible loops in a CFG in
$O(n^2)$. \deleted{Furthermore, experiments on popular programs (such as
\emph{sendmail}, \emph{explorer.exe}, or
\emph{samba}) show that the percentage of functions
containing irreducible loops is low (worst-case is $11\%$ for
\emph{sendmail}, best-case is $0.15\%$ for
\emph{explorer.exe}).} While the complexity of the node-splitting
algorithm is not reported, the algorithm is
meant to be executed only on irreducible loops, which usually
constitute a small part of the analysed program. 
\deleted{In this paper, we assume that an existing node-splitting algorithm
is executed on the CFG before our algorithm.}

\subsubsection{Loop to DAG (Algorithm~\ref{alg:dag})}

\begin{figure}[tbp]
\footnotesize
  \begin{center}
    \subfloat[DAG for loop $l_{b_1}$]{
      \begin{tikzpicture}[scale=0.8]
        \tikzset{vertex/.style = {shape=circle,draw,minimum size=1.5em}}
        \tikzset{edge/.style = {->,> = latex'}} 
        \node (i) at (0,4) {};
        \node[vertex] (a) at (0,3) {$b_1$}; 
        \node[vertex] (b) at (0,1.5) {$L_{b_2}$}; 
        \node[vertex] (c) at (0,0) {$b_3$}; 
        \node[vertex] (d) at (-1.5,3) {$exit$}; 
        \node[vertex] (e) at (-1.5,0) {$next$}; 
        \node[vertex] (f) at (1.5,1.5) {$b_6$}; 
        \draw[edge] (i) to (a); 
        \draw[edge] (a) to (b); 
        \draw[edge] (b) to (c);
        \draw[edge] (c) to (e); 
        \draw[edge] (a) to (d); 
        \draw[edge, bend left] (a) to (f); 
        \draw[edge, bend left] (f) to (c);
      \end{tikzpicture}
      \label{fig:dag1}    
    }
    \hspace{1cm}
    \subfloat[DAG for loop $l_{b_2}$]{
      \begin{tikzpicture}[scale=0.8]
        \tikzset{vertex/.style = {shape=circle,draw,minimum size=1.5em}}
        \tikzset{edge/.style = {->,> = latex'}} 
        \node (i) at (0,2.5) {};
        \node[vertex] (a) at (0,1.5) {$b_2$}; 
        \node[vertex] (b) at (0,0) {$b_4$}; 
        \node[vertex] (c) at (1.5,1.5) {$exit$}; 
        \node[vertex] (d) at (-1.5,0) {$next$}; 
        \draw[edge] (a) to (b); 
        \draw[edge] (a) to (c); 
        \draw[edge] (b) to (d); 
        \draw[edge] (i) to (a);
      \end{tikzpicture}
      \label{fig:dag2}
    }
    \hspace{1cm}
    \subfloat[CFT for the body of loop $l_{b_1}$]{
      \begin{tikzpicture}[scale=0.8]
        \tikzset{vertex/.style = {shape=rectangle,draw,minimum size=1.5em}}
        \tikzset{edge/.style = {->,> = latex'}}
        \node[vertex] (a) at  (-1,0) {$b_1$};
	\node[vertex] (b) at  (0,0) {$Alt$};
	\node[vertex] (d) at  (1,0) {$b_3$};
        \node[vertex] (c) at  (0,1.5) {$Seq$};
	\node[vertex] (e) at (-1,-1.5) {$b_6$};
	\node[vertex] (f) at (1,-1.5) {$L_{b_2}$};
	\draw[edge] (c) to (a);
	\draw[edge] (c) to (b);
	\draw[edge] (c) to (d);
	\draw[edge] (b) to (e);
	\draw[edge] (b) to (f);
      \end{tikzpicture}
      \label{fig:tree1}
    }

    \caption{From loop to DAG and CFT}
  \end{center}
\end{figure}
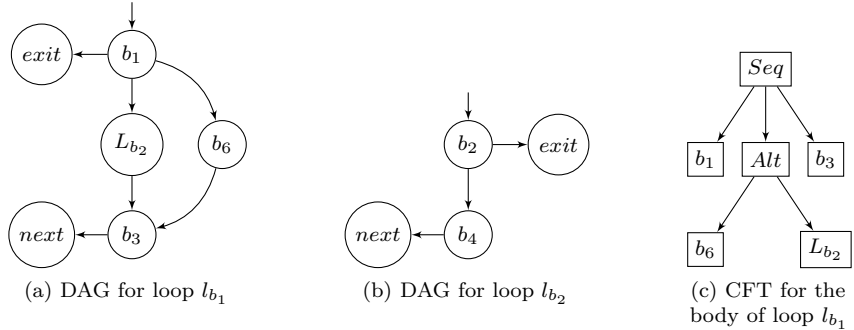

The DAG produced for a loop $l_h$ represents its body. In this DAG, inner
loops are replaced by \emph{hierarchical nodes}, which themselves
correspond to separate DAGs. For instance, Figure~\ref{fig:dag1} shows
the DAG produced for loop $l_{b_1}$ (the construction steps and the
meaning of nodes $exit$ and $next$ are detailed below). $L_{b_2}$ is a
hierarchical node representing loop $l_{b_2}$. The DAG produced for loop
$l_{b_2}$ is shown in Figure~\ref{fig:dag2}. In the remainder of this
section, we use the example of Figure~\ref{fig:dag1} to illustrate
Algorithm~\ref{alg:dag}. 

Algorithm~\ref{alg:dag} constructs the DAG corresponding to a loop
$l_h$. At line 2 the algorithm adds all nodes immediately contained in
$l_h$ to the DAG nodes. Any edge in the CFG between these nodes is
added to the DAG edges (line \substituted{5}{3}). In our example, this
corresponds to nodes $b_1$, $b_3$, $b_6$ and to edges $(b_1,b_6)$ and
$(b_6,b_3$).

Virtual \emph{exit} and \emph{next} nodes are created to represent,
respectively, transferring control to the next
iteration and exiting the loop (lines 4 and 5). For any back-edge $(b_i, b_j)$ of $l_h$ in
the CFG, we add a corresponding edge in the DAG, from $b_i$ to the
virtual next node (line 6). Similarly, for any exit edge $(b_i,b_j)$ of
$l_h$ in the CFG we add a corresponding edge in the DAG, from $b_i$ to
the virtual exit node (line 7).  In our example, we have an edge
$(b_1, \mathsf{exit})$ and an edge $(b_3, \mathsf{next})$.

Inner loops are handled by the \emph{for} in lines 9--14. For each loop
$l_{h^{\prime}}$ directly in $l_h$, we create a \emph{hierarchical node}
$L_{h^{\prime}}$. For each exit edge $(b_i, b_j)$ of $l_h^{\prime}$, an
edge $(L_{h^{\prime}}, b_j)$ is created (line 11) and for each entry
edge $(b_i, b_j)$, an edge $(b_i, L_{h^{\prime}})$ is created (line
12). In our example, a hierarchical node $L_{b_2}$ is created to
represent the loop $l_{b_2}$ (which is directly in loop $l_{b_1}$) and
we also create edges $(b_1,L_{b_2})$ and $(L_{b_2},b_3)$.

We assumed in Section~\ref{sub:ccfg} that the whole CFG is the body of a
(fictive) loop $\top$. Therefore, the whole CFG can also be transformed
into a DAG using Algorithm~\ref{alg:dag}. It produces a hierarchy of
DAGs corresponding to the CFG containing only reducible loops. 

Note that similar algorithms have been proposed
in~\cite{yakdan2015gotos}. However, the most notable difference between
the work presented in~\cite{yakdan2015gotos} and our approach, is that
while our transformation may not preserve the semantics of the program,
we guarantee that it does not decrease the execution time. On the
contrary, the method proposed in~\cite{yakdan2015gotos} guarantees the
preservation of the program semantics, but not the execution time.

\begin{algorithm}[tb]
\caption{Loop to DAG}.
\label{alg:dag}
\footnotesize
\begin{algorithmic}[1]
	\Function{DAG}{$G=<\mathcal{B}, \mathcal{E}>,l_h \in L_G \cup \top$}
	\State $\mathcal{B}_d = \{n | l_h\text{ immediately contains }n\}$
	\State $\mathcal{E}_d = \{(b_i,b_j) | b_i \in \mathcal{B}_d
        \wedge b_j \in \mathcal{B}_d \wedge (b_i, b_j) \in \mathcal{E}\}$
	\State $n \gets$ new virtual node (next)
	\State $e \gets$ new virtual node (exit)
	\State $sn \gets \{(b_i, n) | \exists b_j$, $ (b_i, b_j)$ back-edge of $l_h\}$
	\State $se \gets \{(b_i, e) | \exists b_j$, $ (b_i, b_j)$ exit-edge of $l_h\}$
	\State $(v, i, o) \gets (\emptyset, \emptyset, \emptyset)$
		\For{each loop $l_{h^{\prime}}$ directly in $l_h$}
			\State $L_{h^{\prime}} \gets$ new hierarchical node
			\State $i_{h^{\prime}} \gets \{(L_{h^{\prime}}, b_j) | \exists b_j$, $ (b_i, b_j)$ exit-edge of 
\substituted{$L_{h^{\prime}}$}{$l_{h^{\prime}}$} $\}$
			\State $o_{h^{\prime}} \gets \{(b_i, L_{h^{\prime}}) | \exists b_i$, $ (b_i, b_j)$ entry-edge of 
\substituted{$L_{h^{\prime}}$}{$l_{h^{\prime}}$} $\}$
			\State $(v, i, o) \gets (v \cup \{L_{h^{\prime}}\},i \cup i_h^{\prime},o \cup o_h^{\prime})$
		\EndFor
	\State $d \gets <\mathcal{B}_d \cup v\ \cup\{n, e\},
        \mathcal{E}_d \cup i \cup o \cup \{sn, se\}>$
	\State \Return $(d, n, e)$
	\EndFunction
\end{algorithmic}
\end{algorithm}

\subsubsection{Tree construction (Algorithm~\ref{alg:processing})}

First, we introduce the notion of \emph{forced passage nodes}, upon which the recursive structure
of our algorithm relies. Intuitively, these correspond to the set of nodes that appear in every
path to the end node of a DAG.
\begin{mydef}
  Let $\mathcal{D}$ a DAG. Let \textnormal{start} the start node and
  \textnormal{end} an exit node of $\mathcal{D}$. The set of
  \emph{forced passage nodes} of $\mathcal{D}$ towards $e$,
    denoted $forced(D,e)$, is defined as:
  \[ forced(D,e)=\{n \in \mathcal{D} | \mathsf{start}~\dom~n \wedge n~\dom~\mathsf{end}\} \setminus \mathsf{start}\]
\end{mydef}

The function \Pp\ described by Algorithm~\ref{alg:processing} builds
recursively a control-flow tree from a DAG. Notice that this function
takes as arguments a $\mathsf{start}$ node and an $\mathsf{end}$
node. This is because in some cases it is useful to build the
control-flow tree representing paths between two arbitrary nodes that
are different from the entry and exit nodes of the DAG (see the
different recursive calls in the algorithm for details).

Function \Pp\ returns a $\Seq$ node. The list of children for this
$\Seq$ node is contained in variable $\mathsf{ch}$. We will call this
$\Seq$ node the \emph{current} sequential node.

We denote as $\mathcal{N}$ the set of forced passage nodes towards
$\mathsf{end}$. In the while loop (lines 7 to 19), the algorithm goes through
$\mathcal{N}$ in reverse dominance order (i.e from the $\mathsf{end}$ to the
$\mathsf{start}$). Since we must pass through all nodes in $\mathcal{N}$, it is
clear that each node in $\mathcal{N}$ must be a leaf child of the
current sequential node (line 18). 
As an example, consider the tree obtained from the example of
Figure~\ref{fig:dag1}, which is represented in Figure~\ref{fig:tree1}. During
each iteration of loop $l_{b_1}$, we are forced to pass through $b_1$ and $b_3$, so
$\mathcal{N}=\{b_1,b_3\}$.  Therefore, the control-flow tree has a
$\Seq$ node as root, with children $b_1$ and
$b_3$, as well as an $\Alt$ node whose construction is explained below. 

If there exists multiple possible paths between two adjacent
forced passage nodes (line 10) then an $\Alt$ node must
be added to the $\mathsf{ch}$ list. We construct a tree for each
possible predecessor by recursively calling \Pp, and the
$\Alt$ node contains these trees as children (lines 13 to
15). In our example, the node $b_3$ has two
predecessors, $L_{b_2}$ and $b_6$.  The control-flow trees corresponding
to these two predecessors are respectively $\Leaf(L_{b_2})$ and
$\Leaf(b_6)$.

In lines 20 to 25, the algorithm deals with inner loops. Inner loops
have previously been added to the $\mathsf{ch}$ list as hierarchical
$\Leaf$ nodes. Here, they are replaced by control-flow trees
representing these loops. Such a tree is composed of two parts, in
sequence. The first part is the loop body (line 22), representing all
the iterations of the loop. The second part is the loop exit
$\mathsf{ex}$ (line 23), which represents the paths from the last
execution of the loop header, to the loop exit.  For instance, in
Figure~\ref{fig:treefull} the sub-tree depicted in gray replaces the
hierarchical node $L_{b_2}$. The left part of this sub-tree corresponds
to the body of loop $l_{b_2}$,
while the right part (below the dashed edge) corresponds to the exit of
loop $l_{b_2}$.

We note that in the algorithm, sometimes a single basic block can be
represented by several $\Leaf$ nodes. When such duplication occurs, we
rename the duplicated basic block(s) such that each $\Leaf$ node has an
unique label. This guarantees that two different paths in the tree are
always identified by different sequences of $\Leaf$ nodes.

\begin{algorithm}
\caption{DAG to control-flow tree}
\label{alg:processing}
\footnotesize
\begin{multicols}{2}
\begin{algorithmic}[1]
  \Function{\Pp}{$\mathcal{D}, \mathsf{start}, \mathsf{end}$}
    \State $\mathsf{ch} \gets \emptyset$
    \State $\mathcal{N} \gets$ $forced(\mathcal{D},end)$
    \If{$\mathsf{start}$ has no predecessors}
      \State $\mathsf{ch} \gets \{\Leaf(\mathsf{start})\}$
    \EndIf
    \While{$\mathcal{N} \ne \emptyset$}
      \State Pick $c$ from $\mathcal{N}$ such that $\forall n \in \mathcal{N}, n~\dom~c$
      \State $\mathcal{N} \gets \mathcal{N} \setminus c$
      \If{$c$ has at least 2 predecessors}
         \State $\mathsf{br} \gets \emptyset$
	 \State $ncd \gets$ imm. dominator of $c$ in $\mathcal{N}$
	 \For{$p$ in $predecessors(c)$}
	   \State $\mathsf{br} \gets \mathsf{br} \cup \Pp(\mathcal{D}, ncd, p)$
	 \EndFor
	 \State $\mathsf{ch} \gets \mathsf{ch} \cup \Alt(\mathsf{br})$
      \EndIf
      \State $\mathsf{ch} \gets \mathsf{ch} \cup \Leaf(c)$
    \EndWhile
    \For{all $\Leaf(c) \in \mathsf{ch}$, $c$ representing $l_h$}
      \State $(\mathcal{D}^{\prime}, n, e) \gets \mathsf{DAG}(l_h)$
	  \State $\mathsf{bd} \gets \Pp{\mathcal{D}^{\prime}}, h, n)$
	  \State $\mathsf{ex} \gets \Pp{\mathcal{D}^{\prime}}, h, e)$
      \State Replace $\Leaf(c)$ by $\MyLoop(h, \mathsf{bd}, x_h, \mathsf{ex})$ 
    \EndFor
    \State \Return \Seq($\mathsf{ch}$)
  \EndFunction
\end{algorithmic}
\end{multicols}
\end{algorithm}


\subsection{Execution paths in CFG and Control-flow Tree}

We will now establish a correspondence between CFG execution paths and
tree execution paths.  This subsection contains the general idea and
definitions. For a complete proof, see Appendix~\ref{app:CFG-CFT}.

First, we denote $\gpaths{G, e}$ the function that, given a graph $G$
and a node $e$, returns the set of execution paths
$\{p_1, \ldots, p_k\}$ from the graph entry to the node $e$.

Second, a \emph{tree execution path} is defined as a sequence of leaf
nodes of the tree. We use the same notation for paths in the CFG and for
paths in the tree, with the obvious correspondence between leaf nodes
and basic blocks. The function $\tpaths{t}$ returns the set of tree
execution paths of control-flow tree $t$. It is defined as follows:

\begin{mydef}
  Let $t$ be a Control-flow tree. The set of feasible execution paths of
  $t$, denoted $\tpaths{t}$, is defined inductively as follows:

  \small
  \begin{align*}
    \tpaths{\Leaf(b)}&=\{b\}\\
    \tpaths{\Seq(t_1,\ldots,t_n)}&=\{p|\exists
    p_1\in\tpaths{t_1},\ldots,\exists p_n\in\tpaths{t_n},p=p_1\cat \ldots \cat p_n\}\\
    \tpaths{\MyLoop(h,t_b,n,t_e)}&=\{p|\exists p_1,\ldots,p_n\in\tpaths{t_b},\exists p_e \in \tpaths{t_e},p=p_1 \cat \ldots \cat p_n \cat p_e\}\\
    \tpaths{\Alt(t_1,\ldots,t_n)}&=\bigcup_{1\leq i\leq n} \tpaths{t_i}
  \end{align*}
\end{mydef}

Let us denote $\mathcal{D}_s$ and $\mathcal{D}_e$ respectively the
\emph{start} and \emph{exit} nodes of DAG $D$. Let $G_e$ denote the exit
node of $G$.  The following theorem states the correctness of our
translation from a CFG to a Control-flow Tree: any execution path in the
CFG is also an execution path in the corresponding Control-flow
Tree. However, some paths that are valid in the tree may not be valid in
the CFG, therefore, the two representations are not equivalent. Still,
this is \emph{safe}, since the presence of additional paths in the CFT
can only lead to an \emph{over-approximation} of the WCET.

\begin{theorem}
  \label{th:path-validity}
  Let $G$ be a CFG. Let $\mathcal{D}=DAG(G,\top)$ and let
  $t=\Pp(\mathcal{D}, \mathcal{D}_s, \mathcal{D}_e)$. We have:
  \[ \gpaths{G,G_e} \subseteq \tpaths{t}\]
\end{theorem}

\begin{proof}
  See Appendix~\ref{app:CFG-CFT} for details.
\end{proof}


%% file: abstract-wcet.tex
\section{Context-sensitive execution time}
\label{sub:context-annotation}

We now enrich the control-flow tree with \emph{context annotations}
designed to represent the result of extra-CFG analyses, that will help
us reduce the pessimism in WCET estimation.

\subsection{Context annotations}
A context annotation constrains the conditions under which a
sub-tree can be executed. In this work, annotations only represent
constraints related to loops, which is usually the main source of WCET
variability. Note that with IPET-based approaches, this information
would be represented by an ILP constraint.  We will detail the role of
context annotations in parametric WCET in
Section~\ref{sec:parametric}.

\begin{mydef}
  A context annotation is a tuple $(t, l, m)$, where $t$ is a tree,
  $l$ refers to an external loop (i.e., $l =\MyLoop(h, t_b, x_h, t_e)$
  is a loop such that $t$ is contained within the loop body $t_b$), and $m$
  is the maximum number of times $t$ can be executed each time $l$ is
  entered. The null annotation is denoted by $(t, \top, \infty)$.

  Let $\ann(t)$ be the annotation on the root of tree $t$ and let
  $\annSet(t)$ be the set of annotations on all nodes in $t$
  (including root $t$).

  We define $\occ(\mathcal{P}, p)$ $\occ: 2^{\mathcal{P}} \times
  \mathcal{P} \rightarrow \mathbb{N}$ as the function that returns the
  number of occurrences of any path in $\mathcal{P}$ inside path $p$.

  Let $t$ be a control-flow tree with context annotations. The previous
  definition of feasible execution paths is altered as follows:
  \begin{align*}
    \tpaths{\MyLoop(h,t_1,n,t_2)}=\{& p|\exists p_1,\ldots,p_n\in\tpaths{t_1},
    p_e \in \tpaths{t_2},p=p_1\cat\ldots\cat p_n \cat p_e,\\
    &\forall (t, l_h, m) \in \annSet(t_1), \occ(\tpaths{t}, p) \le m \}
  \end{align*}
\end{mydef}

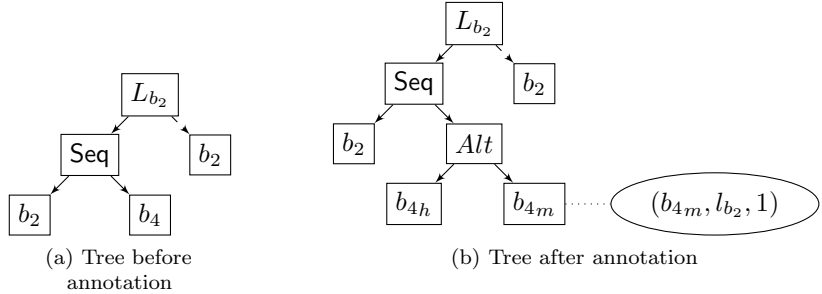
\begin{figure}[tbp]
  \begin{center}
    \subfloat[Tree before annotation]{
      \begin{tikzpicture}[scale=0.8]
        \tikzset{vertex/.style = {shape=rectangle,draw,minimum size=1.5em}}
        \tikzset{edge/.style = {->,> = latex'}}
        \tikzset{edge_exit/.style = {->,> = latex', dashed}}
	\node[vertex] (g) at (1, -1) {\substituted{$\MyLoop(b_2)$}{$L_{b_2}$}};
	\node[vertex] (h) at (2, -2) {$b_2$};
	\node[vertex] (i) at (0, -2) {$\Seq$};
	\node[vertex] (j) at (-1, -3) {$b_2$};
	\node[vertex] (k) at (1, -3) {$b_4$};
	\draw[edge_exit] (g) to (h);
	\draw[edge] (g) to (i);
	\draw[edge] (i) to (j);
	\draw[edge] (i) to (k);
      \end{tikzpicture}
      \label{fig:pre}
    }
    \hspace{1cm}
    \subfloat[Tree after annotation]{
      \begin{tikzpicture}[scale=0.8]
        \tikzset{vertex/.style = {shape=rectangle,draw,minimum size=1.5em}}
        \tikzset{annotation/.style = {shape=ellipse,draw,minimum size=1.5em}}
        \tikzset{edge/.style = {->,> = latex'}}
        \tikzset{edge_exit/.style = {->,> = latex', dashed}}
        \tikzset{edge_annot/.style = {-,> = latex', dotted}}
	\node[vertex] (g) at (1, -1) {\substituted{$\MyLoop(b_2)$}{$L_{b_2}$}};
	\node[vertex] (h) at (2, -2) {$b_2$};
	\node[vertex] (i) at (0, -2) {$\Seq$};
	\node[vertex] (j) at (-1, -3) {$b_2$};
	\node[vertex] (k) at (1, -3) {$Alt$};
	\node[vertex] (hit) at (0, -4) {${b_4}_h$};
	\node[vertex] (miss) at (2, -4) {${b_4}_m$};
	\node[annotation] (ann) at (5, -4) {$({b_4}_m, l_{b_2}, 1)$};
	\draw[edge_annot] (miss) to (ann);
	\draw[edge_exit] (g) to (h);
	\draw[edge] (g) to (i);
	\draw[edge] (i) to (j);
	\draw[edge] (i) to (k);
	\draw[edge] (k) to (hit);
	\draw[edge] (k) to (miss);
      \end{tikzpicture}\label{fig:post}
    }
    \end{center}
    \caption{Context annotations}
\end{figure}

We motivate the need to represent context-sensitive information
by using two examples. First, let us consider a triangular loop: a \emph{for}
loop $i=1..10$, containing an inner \emph{for} loop $j=i..10$. The
maximum iteration count for each loop considered separately is $10$,
but the inner loop body can be executed at most $\sum_{i=1}^{10}i$
times. Knowing this information will enable us to produce a tighter WCET
estimation. To model this example, we have a $\Leaf(b)$ node
representing the block inside the inner loop.  This node has an
annotation $(\Leaf(b), l_{outer}, 55)$ where $l_{outer}$
represents the outer loop. This annotation represent the fact that,
due to the triangular loop, the block $b$ can be executed at most
$\sum_{j=1}^{10}j = 55$ times in a complete execution of $l_{outer}$.

As a second example, we consider the instruction cache analysis by
categorization. In this approach, blocks can be categorized as
\emph{persistent} with respect to a loop (for the sake of simplicity, we
assume that each basic block matches exactly a cache block), meaning
that the block will stay in the cache during the whole execution of the
loop (only the first execution results in a cache miss). For instance,
in the control-flow tree of Figure~\ref{fig:pre}, let us assume that the
block corresponding to $\Leaf(b_4)$ is persistent. For every complete
execution of loop $l_{b_2}$, $b_4$ can only cause a cache miss
once. Thus the execution time of $b_4$ must account for the cache miss
only once per complete execution of loop $l_{b_2}$. To model this
example, we proceed in two steps. First, we modify the CFT by splitting
the block $b_4$ from Figure~\ref{fig:pre} into two (virtual) leaves,
representing respectively the cache hit and cache miss cases. This is
shown in Figure~\ref{fig:post}: $\Leaf({b_4}_m)$ corresponds to the miss
and $\Leaf({b_4}_h)$ to the hit. Then, we add an annotation
$({b_4}_m, l_{b_2}, 1)$ to represent the fact that ${b_4}_m$ can be
executed only once per execution of loop $l_{b_2}$.

Due to context annotations, some structurally feasible paths are
now unfeasible. As an example, in the tree of Figure~\ref{fig:post},
path $\{b_2 , {b_4}_m , b_2 , {b_4}_m , b_2\}$ is feasible if we
ignore annotations. However, taking context annotations into account,
this path it is not.

Context annotations are intended to be a generic tool to model various
WCET-related effects (hardware, and software), therefore the exact way
to generate those annotations will depend on the effect we want to
model (and on the underlying analysis). Furthermore, as shown
  with the cache example above, it may be necessary to modify the CFT
  to represent some constraints. In the future, we might use other CFT
  transformations to represent other types of constraints (not
  necessarily only duplication).

\subsection{Abstract WCET}
\label{sec:abstract}

Due to context annotations, the WCET of a segment of code that is
executed iteratively can vary at each iteration. We introduce the
concept of \emph{abstract WCET} to represent the set of WCETs associated
with a tree node. Abstract WCETs are defined using \emph{multi-sets}, a
generalization of sets where multiple instances of the same element are
allowed. The number of instances of some element $k$ in the multiset is
denote $m(k)$ and called its \emph{multiplicity}. In our context, we
consider that the smallest element of the multiset has an implicit
infinite multiplicity. We recall below some definitions on multi-sets:

\begin{mydef}
\label{def:multiset}
Let $\mathbb{N}^\#$ denote the set of multi-sets over $\mathbb{N}$. Let
$\eta,\eta^{\prime}\in\mathbb{N}^\#$ and let $n\in\mathbb{N}$. The following operations are
defined on multi-sets:
  \begin{itemize}
  \item $\eta [n]$, denotes the $(n+1)$-th greatest element of $\eta$, i.e.
     $|\{k|k\in\eta,k>\eta[n]\}|\leq n<
    |\{k|k\in\eta,k>\eta[n]\}|+m(\eta[n])$. For instance, if
    $\eta=\{4,3,3\}$ then $\eta[0]=4$, $\eta[1]=\eta[2]=3$, ...;
  \item $\eta|_n$ denotes the multi-set that contains the $n$ greatest
    elements of $\eta$ (i.e. $\eta[0],\ldots,\eta[n-1]$) and an infinite number of zeros;
  \item $\eta \uplus \eta^{\prime}$ is a modified version of the
    traditional multi-set sum, which we will denote
    $\uplus_{trad}$. Like $\uplus_{trad}$, $\uplus$ sums
    multiplicities. The difference is as follows. Let $min_{\eta}$,
    $min_{\eta^{\prime}}$ denote respectively the smallest elements of
    $\eta$ and $\eta^{\prime}$. Then, we have:
    $\eta\uplus\eta^{\prime}=\eta\uplus_{trad}\eta^{\prime}\backslash\{k|k\leq
    max(min_{\eta}, min_{\eta^{\prime}})$.
    So for instance, $\{8,8,4\}\uplus\{9,8,3,2\}=\{9,8,8,8,4\}$;
  \item $\eta \otimes k$ denotes the multi-set for which each member has $k$
    times the multiplicity it has in $\eta$;
  \item $\eta^{\prime\prime} = \eta \oplus \eta^{\prime}$ is the
    multi-set such that: $\forall i\in\mathbb{N}$, $\eta^{\prime\prime}[i] = \eta[i] + \eta^{\prime}[i]$.
  \end{itemize}
\end{mydef}

The notion of \emph{abstract WCET} is now defined as follows:

\begin{mydef}
  For any tree $t$, its abstract WCET is a pair $\alpha = (l, \eta)$,
  where $l$ is a loop and $\eta$ is a multi-set over $\mathbb{N}$. The
  presence of an integer $n$ in $\eta$ means that the code associated
  with $t$ may have an execution time $n$, but only once, each time $l$
  is entered.
\end{mydef}

For instance, in our cache example from Figure~\ref{fig:post}, the
abstract WCET computed for the \Alt\ node would be
$(l_{b_2},\{\wcet{{b_4}_m},\wcet{{b_4}_h},\wcet{{b_4}_h},\wcet{{b_4}_h},...\})$,
meaning that the WCET of that node is $\wcet{{b_4}_m}$ for the first
iteration of loop $l_{b_2}$ and then it is $\wcet{{b_4}_h}$ for all
subsequent iterations of the loop. Note that, if we exit and re-enter
the loop, the WCET of the \Alt\ node will again be $\wcet{{b_4}_m}$,
then $\wcet{{b_4}_h}$, $\wcet{{b_4}_h}$, etc.

The abstract WCET for an expression $t \in \mathcal{T}$ is computed by
applying the evaluation function $\gamma : \mathcal{T} \rightarrow
L_G\times \mathbb{N}^\#$, defined below, using helper function $\omega(t)$:

\begin{center}
  $\gamma(t) = (l, \eta)$\quad\quad
where $l = l_1 \sqcap l_2$, $\eta = \eta_1|_n$, $(l_1, \eta_1) =
\omega(t)$ and $(t, l_2, n) = \ann(t)$.
\end{center}

\added{$\omega(t)$
  computes the abstract WCET without considering the annotation on the
  root node of $t$, and then $\gamma(t)$ computes the abstract WCET
  resulting from the application of the annotation over $t$ (if any).}
Notice that, if no annotation is defined over $t$, then
$\ann(t) = (t, \top, \infty)$; as a consequence $l \sqcap \top = l$,
$\eta|_\infty = \eta$, and $\gamma(t) = \omega(t)$.
\deleted{Basically, $\omega$ computes the abstract WCET without
  considering the annotation on the root node of $t$, and $\gamma$
  computes the abstract WCET resulting from the application of the
  annotation.}

We now define $\omega(t)$ for the different cases. First, when
$t=\Leaf(b)$, the WCET of the basic block $b$ is repeated an infinite
number of times. In formula:

\begin{center}
  $\omega(t) = (\top, \{\wcet{b}\} \otimes \infty)$
\end{center}

The idea behind the processing of \Alt\ nodes is based on the
following observation: the worst-case scenario for multiple executions
of the \textsf{Alt} node may involve execution of different
children. Therefore, we need to merge the multi-sets resulting from
$t_1, \dots, t_n$. In formula, when $t = \Alt(t_1, \dots, t_n)$:

\begin{center}
  $\omega(t) = (l_1 \sqcap \dots \sqcap l_n, \eta_1 \uplus \dots \uplus \eta_n)$\quad\quad
  where $(l_1, \eta_1)=\gamma(t_1)$, \dots, $(l_n, \eta_n)=\gamma(t_n)$.
\end{center}

\begin{example}
  Let us consider an \Alt\ node with two children $t_1$ and $t_2$,
  such that $\gamma(t_1) = (l, \{5,4,2,1\})$ and
  $\gamma(t_2) = (l, \{6, 2\})$. The first time the \Alt\ node is
  executed, the WCET will be $6$ (from $t_2$), the second time it will
  be $5$ (from $t_1$), then $4$, and so on.  As such, we compute the
  abstract WCET for the \Alt\ node by taking the union of the
  multi-set components of the two children abstract WCET.  Therefore,
  in our example,
  \substituted{$\omega(t) = (l, \{6, 5, 4, 2, 1\})$}
              {$\omega(t) = (l, \{6, 5, 4, 2, 2\})$}.
\end{example}

When $t = \Seq(t_1, \dots, t_n)$, 
we make the following observation: for any
$n$, the worst-case time for $n$ executions of the \Seq\ node is equal
to the worst-case time for $n$ executions of $t_1$ plus the worst-case
time for $n$ executions of $t_2$ and so on. In formula:
\begin{center}
  $\omega(t) = (l_1 \sqcap \dots \sqcap l_n, \eta_1 \oplus \dots\oplus \eta_n)$\quad\quad
  where $(l_1, \eta_1)=\gamma(t_1)$, \dots, $(l_n, \eta_n)=\gamma(t_n)$.
\end{center}

\begin{example}
  Let us consider a \Seq\ node with two children $t_1$ and $t_2$, such
  that $\gamma(t_1) = (l, \{5,4\})$ and $\gamma(t_2) = (l,
  \{2,1\})$. The first time the \Seq\ node is executed, its WCET will
  be $5+2=7$, the second time it will be $4+1=5$.  As such, we compute
  the abstract WCET for the \Seq\ node by adding elements of
  corresponding ranks. In the example, $\omega(t) = (l, \{7, 5\})$.
\end{example}

When $t = \MyLoop(h, t_1, x_h, t_2)$, \added{let $(l_1, \eta_1) = \gamma(t_1)$ and $(l_2, \eta_2) = \gamma(t_2)$}. 
Two different cases must be
considered\footnote{Notice that, by definition of context annotation,
  it is not possible to have $l_h \equiv l_2$.}. If $l_h$ is the loop
component of the abstract WCET of $t_1$ (case \substituted{$l_h=l$}{$l_1 \equiv l_h$}), then the
execution time of $t_1$ is a fixed value. In this case, the worst-case time for
one execution of the \MyLoop\ node is always the worst case execution
time for $x_h$ executions of the loop body $t_1$.

Otherwise, \substituted{$l$}{$l_1$} represents a loop that contains the currently processed
\MyLoop\ node. As in the previous case, the worst-case execution time
for one execution of the \MyLoop\ node is the worst-case execution time
for $x_h$ executions of the loop body $t_1$. However, since \substituted{$l$}{$l_1$} refers
to an outer loop, successive executions of the \MyLoop\ node yield
different execution times, and these are summed together in groups of
$x_h$ elements.

To summarize, in formula: 
\begin{center}
  \begin{equation*}
    \omega(t) = \begin{cases}
      (l_2, (\{\sum_{i=0}^{x_h-1} \eta_1[i]\} \otimes +\infty) \oplus \eta_2) & \text{if } l_h \equiv l_1 \\
      (l_1 \sqcap l_2, \eta \oplus \eta_2) & \text{otherwise }
    \end{cases}
  \end{equation*}
  where $(l_1, \eta_1)=\gamma(t_1)$ and $(l_2, \eta_2)=\gamma(t_2)$ and
  $\eta[i]=\sum_{j=i\cdot x_h}^{i \cdot x_h+x_h-1}\eta_1[j]$.
\end{center}

\begin{example}
  Let $\gamma(t_1)=(l_h, \{5,4,3\})$ \added{(case $l_h \equiv l_1$)}, let the loop bound $x_h=2$ and let
  $t_2$ be empty. Then the execution time for one execution of the loop
  is always $5+4=9$ (the sum of the $x_h$ first ranks of the multi-set)
  and we have $\omega(t)=(\top, \{9\} \otimes \infty)$.
\end{example}

\begin{example}
  Let \substituted{$\gamma(t_1)=(l, \{5,4,3,2\})$}{$\gamma(t_1)=(l_1, \{5,4,3,2\})$ (case $l_h \not\equiv l_1$)}, let the loop bound $x_h=2$ and let
  $t_2$ be empty. Then the first execution of the loop will yield
  execution time $5+4=9$ (the sum of the first $x_h$ ranks of the
  multi-set), while the second execution will yield execution time
  $3+2=5$ (the sum of the subsequent $x_h$ ranks of the multi-set).
  Therefore, \substituted{$\omega(t)=(l, \{9,5\})$}{$\omega(t)=(l_1, \{9,5\})$}.
\end{example}

Notice that we make pessimistic simplifications concerning the loop
component in the computation of $\omega$ and $\gamma$. Consider, the
computation for $t=\Alt(t_1, \dots, t_n)$ for instance. The WCET of
$t_1, \dots, t_n$ may depend on different loops, but keeping track of
all these loops in the WCET of $t$ would be very complex. So, as a
simplification, we only keep track of the greatest lower bound of
these loops (the loop that most immediately contains $t$). This is
also true in other cases. However, this approximation is safe (see the
proof of Theorem~\ref{app:WCET-correctness} for details) and has a low
impact on WCET over-approximation (see Section~\ref{sec:eval}).

\subsection{From abstract to concrete WCET}
\label{sec:awcet-eval}

We will now detail how to evaluate the WCET of a tree $t$
\substituted{ executed $n$ times (where $(l, \eta)=\gamma(t)$)}{inside
  a loop $l$. Suppose that $t$ is executed $n$ times and that its
  abstract WCET is $\gamma(t)=(l, \eta)$.} The execution time for each
individual execution of $t$ depends on the number of times it was
executed after the last time $l$ was entered. Let $e$ be the number of
times $l$ was entered, and let us assume that the $n$ execution of $t$
are distributed uniformly across all $e$ executions of $l$ (this is a
realistic assumption because our computation method ensures that
iterating every loop to the maximum results in the longest execution
time).

\begin{mydef}
  Let $t$ be a control-flow tree and let $\gamma(t)=(l,\eta)$. The
  concrete WCET of $t$ in the scenario where $t$ is executed $n$ times
  and the loop $l$ is executed $e$ times, where $e$ and $n$ are strictly
  positive and $n$ is a multiple of $e$, is computed as: $\sum_{i=1}^n{(\eta \otimes e)[i]}$
\end{mydef}

This definition applies to any node of the tree. To compute the WCET of
a complete program represented by tree $t$, we apply the formula with
$n=e=1$, since we are only interested in one execution of the
program. The WCET of the program is thus computed as
$\sum_{i=1}^1{(\eta \otimes 1)[i]}=\eta[1]$. The following theorem
establishes the soundness of our WCET evaluation method.

\begin{theorem}
  Let $G$ a CFG. Let $\mathcal{D}=DAG(G,\top)$ and let
  $t=\Pp(\mathcal{D}, \mathcal{D}_s, \mathcal{D}_e)$. Let
  $(l,\eta)=\gamma(t)$. We have: $\forall p\in\gpaths{G,G_e}, \wcet{p} \leq \eta[1]$
\end{theorem}

\begin{proof}
  See Appendix~\ref{app:WCET-correctness} for details.
\end{proof}


%% file: parametric.tex
\section{Symbolic computation}
\label{sec:parametric}

In this section we study the problem of computing the abstract WCET of a
tree when some parameters of the tree are unknown (loop bounds for
instance, but not only). We show that, using simple syntactic sugaring,
our definition of $\omega(t)$ produces formulae akin to arithmetic
expressions. Then we rely on existing work on symbolic computation of
arithmetic expressions to simplify abstract WCET formulae. The
simplification step is mainly useful in case of on-line formula
evaluation. It reduces memory overhead (since formulae must be part of
the embedded code) as well as execution time overhead (since formulae
must be evaluated at each task instanciation).

\subsection{Abstract WCET formulae}
\label{sec:wcet-formulae}

First, we introduce several operators on abstract WCET, which act as syntactic
sugar, to be able to express WCET computation as arithmetic
computation.

\begin{mydef}
  Let $t_1$ and $t_2$ be two control-flow trees. We define a set of
  operations on abstract WCET such that:

  \begin{align*}
    \omega(t_1)\plusw \omega(t_2)&= \omega(\Seq(t_1,t_2))\\
    \omega(t_1)\maxw \omega(t_2)&= \omega(\Alt(t_1,t_2))\\
    (\omega(t_1),\omega(t_2),h)\poww{x_h} &=
    \omega(\MyLoop(h,t_1,x_h,t_2))\\
    \annw{\omega(t_1)}{(h,n)} &= \gamma(t_1)\quad \text{(where
    }\ann(t_1)=(t_1,l_h,n))\\
    n\dotw (l,\eta) &= (l,\eta^{\prime}) \quad \text{(where }\forall
    i,\eta^{\prime}[i]=\eta[i]\times n)\\
    \constw{k}&=\{k\}\otimes \infty\
  \end{align*}

\end{mydef}

Furthermore, we let $\zerow\equiv (\top,\constw{0})$. We define
the following grammar to represent the set of formulae $\wformula$
corresponding to the computation of the abstract WCET of a control-flow
tree ($w\in\wformula$):

\[
\begin{array}{lcl}
  w & ::= & const \Choice id \Choice \annw{w}{(h,it)}
  \Choice w\plusw w \Choice w\maxw w \Choice (w,w,b)\poww{it}\\
  h & ::= & b \Choice id\\
  it & ::= & i \Choice id\\
\end{array}
\]

The simplest formula is a constant abstract WCET value ($const\in
(L_G\times N^\#)$). A formula can also be a variable corresponding
to an unknown WCET value ($id$). A formula can also be the sum ($w\plusw
w$), the product ($w\maxw w$) or the repetition of two formulae
($(w,w,b)\poww{it}$).  Finally, a formula can also consist of the
application of an annotation to a formula ($\annw{w}{(h,it)}$). The factor
of a repetition and the factor of an annotation ($it$) can either be a
constant integer value ($i$) or a variable ($id$). The loop header of an
annotation ($h$) can either be a basic block name ($b$) or a variable ($id$).

\subsection{Symbolic values}
\commentaire{This whole section is a complete rework of the 
  paragraphs previously at the end of Section~\ref{sec:wcet-formulae}.}
\label{sec:symb-ex}

As we can see, several elements of these formulae can be symbolic values
(denoted by $id$), i.e. variable parameters: symbolic WCET value ($w$), symbolic
loop iteration bound ($it$), symbolic loop header ($h$). Let us now illustrate how these symbolic values can be used to model various
WCET variation sources. A simple example is the case where the
number of iterations of a loop depends on an input of the
system. The WCET of the loop is statically evaluated to
$(\omega_1,\omega_2,h)\poww{n}$, where $n$ is a symbolic value. The value
of $n$ is computed dynamically and the WCET of the loop is deduced from
this value.

As a second example, we discuss how to perform a modular WCET analysis,
in the case where the program contains a call to a dynamic library. Assume
for instance that the library call is in the $else$ branch of
an $if-then-else$ and that the $then$ branch has a
constant WCET of 5. The WCET is statically evaluated to
$(\top,\{5\}\otimes\infty)\maxw \omega$, where $\omega$ is a symbolic
value. We perform a separate analysis on the different programs the
dynamic library call can correspond to, so we obtain a different WCET
for each possibility. At program execution, we replace $\omega$ by the
WCET corresponding to the library that is actually called and
deduce the program WCET.

As a last example we discuss how to take into account the results of
an instruction cache analysis. Let us consider the execution of a
multi-task system with a non-preemptive scheduler. In such a system,
though the hardware provides no means to consult the exact cache state,
it can be approximated to an abstract cache state using the techniques
of \cite{blackbox}. In some cases, the category of a block, that is to
say whether the execution of the block will result in a miss or in a
hit, depends on the content of the cache at the beginning of the
execution of the task containing it. As a consequence, the block
category cannot be determined statically, however it can be determined
dynamically based on the abstract cache state at the beginning of the
task execution. In Figure~\ref{fig:post}, we have shown how to use
context annotations to model a persistent block. Similarly, to model a
block with a non-static category, we split the block into a $hit$ and a $miss$
alternative, and add annotations on both alternatives. So the WCET
formula for this block will be:
$\annw{(\omega_{hit})}{h,n_1}\maxw\annw{(\omega_{miss})}{h,n_2}$, where
$n_1$ and $n_2$ are symbolic values. At the beginning of
the task execution, we determine the values of $n_1$ and
$n_2$ based on the abstract cache content and deduce the task
WCET. A similar approach can be used to take into account data-cache
analysis and branch prediction.

More generally, we believe that symbolic WCET evaluation is a powerful
generic tool with many potential applications. The focus of this paper
however, is to present the general framework. Potential applications
will be the subject of future work.

Concerning the limitations of our approach, currently we cannot specify
constraints relating different symbolic values, which may prevent some
simplifications in WCET formula. For instance, a single parameter in the
program external context (e.g. the data-cache size) may introduce
several separate symbolic values in the WCET formula (e.g. the WCET of
each basic-block whose WCET is impacted by the data-cache size will
become a symbolic value). Handling such related symbolic values is clearly
also an important topic for future work.

A second limitation is that some extra-CFG analyses information may be
difficult to represent using context-annotations, such as for instance
the results of CCG analysis~\cite{Li1996}.

\subsection{Formula simplification}
\label{sub:simp}

When variables appear in a WCET formula, we cannot reduce the
formula to a constant abstract WCET value. However, in many cases the
formula can be transformed into a simpler, yet equivalent formula. For
instance, we have: $(x\plusw 2\dotw x)\plusw 3\dotw x\plusw y=6\dotw x\plusw y$

Figure~\ref{fig:rewrite} lists all the rewriting rules we use in order
to simplify WCET formulae. Most of them are direct transpositions of
integer arithmetic simplification rules \cite{cohen2003} to the case of
WCET formulae. We make the following comments:
\begin{itemize}
\item We rely on an order relation $\avant$ on formulae,
  so as to ensure that the commutativity rules can only be applied in
  one direction for two given formulae. Classically, the order relation
  is defined based on the syntactic structure of the formulae (see
  e.g. \cite{cohen2003} for details);
\item Distributivity is applied in reverse order and only to factor
  constant terms;
\item Concerning the annotation rewriting rule, the strategy consists in
  reducing the number of annotation applications;
\item Concerning the loop rule, since we have no rule for combining
  loops, we only extract the loop exit tree from the loop;
\item Combination of constant formulae is not detailed here but is
  applied as well. For instance, $(l,\constw{2})\plusw (l,\constw{3})$
  is simplified to $(l,\constw{5})$.
\end{itemize}

\begin{figure}
   
  \begin{minipage}[t]{.48\linewidth}
    \begin{center}

    \emph{Associativity.}
    \begin{align}
      (w_1 \plusw w_2) \plusw w_3 &\mapsto w_1 \plusw w_2 \plusw w_3   \label{rw:assoc-plus1}\\
      w_1 \plusw (w_2 \plusw w_3) &\mapsto w_1 \plusw w_2 \plusw w_3 \label{rw:assoc-plus2}\\
      (w_1 \maxw w_2) \maxw_3 &\mapsto w_1 \maxw w_2 \maxw w_3 \label{rw:assoc-max1}\\
      w_1 \maxw (w_2 \maxw w_3) &\mapsto w_1 \maxw w_2 \maxw w_3 \label{rw:assoc-max2}
    \end{align}
    
    \emph{Commutativity.}
    \begin{align}
      (w_1 \plusw w_2) &\mapsto (w_2 \plusw w_1) \text{ if $w_2 \avant w_1$ } \label{rw:commut-plus}\\
      (w_1 \maxw w_2) &\mapsto (w_2 \maxw w_1) \text{ if $w_2 \avant w_1$ }\label{rw:commut-max}
    \end{align}
    
    \emph{Distributivity.} 
    \begin{align}
      (cst_1& \plusw w_3) \maxw (cst_2 \plusw w_3) \mapsto \nonumber\\
      &(cst_1 \maxw cst_2) \plusw w_3 \label{rw:distributivity}
    \end{align}
    \end{center}
  \end{minipage}
  \hfill
  \begin{minipage}[t]{.48\linewidth}
    \begin{center}

    \emph{Neutral element.}
    \begin{align}
      w_1 \plusw \zerow &\mapsto w_1 \label{rw:neutral-plus}\\
      w_1 \maxw \zerow &\mapsto w_1 \label{rw:neutral-max}
    \end{align}

    \emph{Multiplication.}
    \begin{align}
      0 \dotw w_1 &\mapsto \zerow \\
      (k_i \dotw w_1) \plusw w_1 &\mapsto (k_i + 1) \dotw w_1
    \end{align}
    
    \emph{Annotation.}
    \begin{align}
      \annw{\zerow}{(h,it)} &\mapsto \zerow \\
      \annw{w_1}{(h,it)} \plusw \annw{w_2}{(h,it)} &\mapsto \annw{(w_1
                                                     \plusw w_2)}{(h,it)} \label{rw:annot-plus}
    \end{align}

    \emph{Loop.}
    \begin{align}
      (w_1, w_2, b)\poww{it} \mapsto (w_1, \zerow, b)\poww{it} \plusw
      w_2
      \label{rw:loop2}
    \end{align}
  \end{center}
  \end{minipage}
  \caption{\label{fig:rewrite}Abstract WCET formula rewriting rules}
\end{figure}

Let $\mathcal{R}$ denote the rewriting system consisting of all of these
rewriting rules. Let $w_1,w_2$ two WCET formulae. We write
$w_1\mapsto_{\mathcal{R}} w_2$, or simply $w_1\mapsto w_2$ when $w_1$
rewrites to $w_2$ using a single rule of $\mathcal{R}$. We write
$w_1\rwtrans w_2$ when $w_1$ rewrites to $w_2$ using a sequence of rules
of $\mathcal{R}$. Let $\rho$ denote a \emph{variable mapping}, that is
to say a set of substitutions of the form $id\rightarrow v$ where $id$
is an identifier and $v$ is a value. Let $\rho(w)$ denote the result of
the substitution of variables of $w$ by their values in
$\rho$. We assume that $\rho$ maps identifiers to values of the correct
type, meaning that it maps WCET identifiers to WCET values, loop
identifiers to loop headers and integer identifier to integer values. We
say that $\rho$ is a \emph{complete mapping} with respect to formula $w$
when it maps all variables of $w$ to a value.

\begin{lemma}
  Let $w_1,w_2\in\mathcal{\wformula}$. Let $\rho$ a complete variable
  mapping of $w_1$. We have:

  \[ w_1\rwtrans w_2 \Rightarrow \rho(w_1)=\rho(w_2)\]
\end{lemma}

\begin{proof}
  We must prove that, for each rewriting rule, the formula on the
  left of the rule is equivalent to the formula on the right. Most rules are trivial to prove and rely on arithmetic
  properties on integer multi-sets. We only detail the proof for rules on
  annotations and loops.

  \emph{Rule~\ref{rw:annot-plus}}.
	Let $(l_1, \eta_1)=w_1$ and $(l_2, \eta_2)=w_2$. 
	\begin{align*}
		\annw{(w_1 \plusw w_2)}{(h,it)} &= (l_1 \sqcap l_2, (\eta_1 \plusw \eta_2)|_{it})
		= (l_1 \sqcap l_2, \eta_1|_{it} \plusw \eta_2|_{it}) \\
	    &= (l_1, \eta_1|_{it}) \plusw (l_2, \eta_2|_{it})
              = \annw{w_1}{(h,it)} \plusw \annw{w_2}{(h,it)}
	\end{align*}

	\emph{Rule~\ref{rw:loop2}}.
	Let $(l_1, \eta_1)=w_1$ and $(l_2, \eta_2)=w_2$. 
	
	By definition of the $\omega$ function on \MyLoop\ nodes, we see
        that the computation result for $(w_1, w_2, b)\poww{it}$ is of
        the form $(\langle loop \rangle, \langle expression \rangle
        \plusw \eta_2)$. Therefore, let us define $(l, \eta)$ such that
        $(w_1, w_2, b)\poww{it}$ = $(l, \eta \plusw \eta_2)$.

	If $l_1 = b$ then:
	
	\begin{align*}
		(w_1, w_2, b)\poww{it} &= (l_2, \eta \plusw \eta_2) = (\top, \eta \plusw \constw{0}) \plusw (l_2, \eta_2) = (w_1, \zerow, b)\poww{it} \plusw w_2
	\end{align*}

	If $l_1 \ne b$ then:
	
	\begin{align*}
		(w_1, w_2, b)\poww{it} &= (l_1 \sqcap l_2, \eta \plusw \eta_2) 
          = (l_1 \sqcap l_2, \eta \plusw \constw{0}) \plusw (l_2, \eta_2) \\
	&= (l_1, \eta \plusw \constw{0}) \plusw (l_2, \eta_2) 
	= (w_1, \zerow, b)\poww{it} \plusw w_2.
	\end{align*}
        This concludes the proof.
\end{proof}

The following Lemma states that recursive applications $\mathcal{R}$ to
a given formula $w$ eventually reach a fixed-point and always produce
the same formula $w'$.

\begin{lemma}
  $\mathcal{R}$ is convergent.
\end{lemma}

\begin{proof}
  $\mathcal{R}$ is convergent if it \emph{terminates} and it is
  \emph{confluent}. The reader can refer to \cite{baader98} for more
  detailed definitions and proof strategies that we use here.

  \emph{Termination.}
  We note that for each rule $l\mapsto r$ of $\mathcal{R}$, we have
  either of the following properties:
  \begin{itemize}
  \item Let $op(w)$ denote the sum of the number of operators
    $\plusw,\maxw,\dotw,|$ in $w$. Then, we have $op(l)<op(r)$ (for the
    following rules: distributivity, neutral element, multiplication
    with an integer, annotation);
  \item The number of parenthesis is less in $l$ than in $r$
    (for associativity rules);
  \item $l\avant r$ (for commutativity rules);
  \item Let us extend $op$ by defining
    $op((w_1,w_2,h)\poww{k})=(k+1)*(op(w_1)+op(w_2))$. Then
    $op(l)<op(r)$ (for loop rules).
  \end{itemize}

  Based on these properties, we can define a strict order
  relation $\prec$ on formulae such that, for each rule $l\mapsto r$ we have
  $l\prec r$. As a consequence $\mathcal{R}$ terminates.

  \emph{Confluence.}
  As $\mathcal{R}$ terminates, we only need to prove that its
  overlapping rules are locally confluent. Two rules $l_1\mapsto r_1$ and
  $l_2\mapsto r_2$ overlap
  if there exists a sub-term $s_1$ of $l_1$ (resp. $s_2$ of $l_2$) that
  is not a variable, and a unifier (a term substitution) $u$ such that
  $u(s_1)=u(l_2)$ (resp. $u(s_2)=u(l_1)$). Unification is applied after renaming variables such
  that $Vars(l_1)\cap Vars(l_2)=\emptyset$. For instance, rules
  \ref{rw:assoc-plus1} and \ref{rw:assoc-plus2} overlap: we have
  two different possible sequences of re-writings for formula
  $(w_1\plusw(w_2\plusw w_3))\plusw w_4$:
  \begin{align*}
    (w_1\plusw(w_2\plusw w_3))\plusw w_4 &\mapsto (w_1+w_2+w_3)+w_4 \text{(rule \ref{rw:assoc-plus2})}\\
    &\mapsto w_1+w_2+w_3+w_4 \quad \text{(rule \ref{rw:assoc-plus1})}\\\\
    (w_1\plusw(w_2\plusw w_3))\plusw w_4 &\mapsto w_1+(w_2+w_3)+w_4 \text{(rule \ref{rw:assoc-plus1})}\\
    &\mapsto w_1+w_2+w_3+w_4 \quad \text{(rule \ref{rw:assoc-plus2})}
  \end{align*}
  As both sequences produce the same formula, these overlapping rules
  are locally confluent.

  We do not detail the proof for the remaining overlapping rules, since
  it is very similar to the case we just presented. We only list them
  below:
  
  \begin{align*}
	  &(w_1\maxw(w_2\maxw w_3))\maxw w_4\quad &\text{(\ref{rw:assoc-max1} and \ref{rw:assoc-max2})}\\    
	  &(w_1\plusw w_2)\plusw w_3 \quad \text{if }w_2\avant w_1\quad &\text{(\ref{rw:assoc-plus1} and \ref{rw:commut-plus})}\\
	  &w_1\plusw (w_2\plusw w_3) \quad \text{if } w_3\avant w_2\quad &\text{(\ref{rw:assoc-plus2} and \ref{rw:commut-plus})}\\
	  &(w_1\maxw w_2)\maxw w_3 \quad \text{if } w_2\avant w_1\quad &\text{(\ref{rw:assoc-max1} and \ref{rw:commut-max})}\\
	  &w_1\maxw (w_2\maxw w_3) \quad \text{if } w_3\avant w_2\quad &\text{(\ref{rw:assoc-max2} and \ref{rw:commut-max})}\\
	  &(cst_1 \plusw (w_3\plusw w_4)) \maxw (cst_2 \plusw (w_3\plusw w_4)\quad
	&\text{(\ref{rw:assoc-plus2} and \ref{rw:distributivity})}\\
          &(cst_1 \plusw w_3) \maxw (cst_2 \plusw w_3)\quad \text{if }w_3\avant
	cst_1\vee w3\avant cst_2\quad &\text{(\ref{rw:commut-plus} and \ref{rw:distributivity})}\\
	  &(w_1\plusw\zerow)\plusw w_2\quad &\text{(\ref{rw:assoc-plus1} and \ref{rw:neutral-plus})}\\
	  &w_1\plusw(\zerow\plusw w_2)\quad &\text{(\ref{rw:assoc-plus2} and \ref{rw:neutral-plus})}\\
	  &(w_1\maxw\zerow)\maxw w_2\quad &\text{(\ref{rw:assoc-max1} and \ref{rw:neutral-max})}\\
	  &w_1\maxw(\zerow\maxw w_2)\quad &\text{(\ref{rw:assoc-max2} and \ref{rw:neutral-max})}\\
	  &\annw{w_1}{(h,it)} \plusw \annw{w_2}{(h,it)}\quad \text{if } w_2\avant w_1\quad &\text{(\ref{rw:annot-plus} and \ref{rw:commut-plus})}\\
  \end{align*}
  This concludes the proof. 
\end{proof}

To summarize, we enumerate below the steps of the computation of the WCET of a
program with our approach. Steps 1 to 4 correspond to the computation of
the parametric WCET formula. Steps 5 and 6 correspond to the
computation of the actual WCET for some specific parameter values:
\begin{enumerate}
\item Translate the program CFG to a CFT $t$;
\item Add extra-CFG analyses results as context annotations;
\item Compute $w=\gamma(t)$;
\item Simplify $w$ into $w'$ using rewriting rules;
\item Replace parameters by their values and obtain $w''$, with $w''=(l,\eta)$;
\item Return $\eta[1]$.
\end{enumerate}


%% file: experiments.tex
\label{sec:eval}

\begin{table}
  \centering
  \resizebox{\textwidth}{!}{
  \begin{tabular}{|c||c|c|c|c|}
    \hline
    \emph{Bench} & \emph{Source} & \emph{Parameter} & \emph{Algorithm} &
    \emph{Function}\\ \hline
    matmult  & ML & Matrix size & Matrix multiplication & \emph{Initialize (twice)}\\\hline
    cnt & ML & Matrix size & Matrix sum & \emph{Sum} \\ \hline
    fft      & TB & Number of samples & FFT & \emph{main}\\ \hline
    compress & ML & Data size & Data compression & \emph{main} \\\hline
    lift     & TB & Number of sensors & Factory lift control & \emph{main} \\ \hline
    adpcm    & ML & Trigo. computation steps & ADPCM encoding & \emph{main}\\ \hline
    aes\_enc & TB & Data size & AES encryption & \emph{main}\\ \hline
    powerwindow & TB & Sensor data input size & Car window control & \emph{main} \\ \hline
    fbw      & PB & Task activaction count & fly-by-wire & \emph{main} \\ \hline
    audiobeam & TB & Audio source count & Audio beamforming & \emph{main} \\ \hline
    mpeg2    & TB & Video resolution & MPEG2 decoding & \emph{main} \\ \hline
  \end{tabular}}
  \caption{Benchmarks summary}
  \label{tab:benchs}
\end{table}

The benchmarks we selected for our experiments are summarized in
Table~\ref{tab:benchs}. For each benchmark, we mention its source (ML
for M\"alardalen, TB for TACleBench, or PB for PapaBench), provide a
short description of the kind of algorithm it performs and specify the
function whose WCET is analyzed. We only introduce one parameter per
benchmark because 
precision is independent of the number of parameters in our
approach. The analyses have been executed on a PC with an Intel core
i5 3470 at 3.2 Ghz, with 8 Gb of RAM.  Every benchmark has been
compiled with ARM crosstool-NG 1.20.0 (gcc version 4.9.1) with
\textsf{-O1} optimization level.

\begin{table}[tbp]
\centering
\resizebox{\textwidth}{!}{
\begin{tabular}{| c || c || c | c || c | c | c || c | c | c | c |}
	\hline
	 \multicolumn{2}{|c||}{} & \multicolumn{2}{c||}{\emph{Formula
                                   size}} &
                                            \multicolumn{3}{c||}{\emph{Time
                                            (ms)}} & \multicolumn{4}{c|}{\emph{Pessimism (\%)}}  \\
	\hline
	\emph{Bench} & \emph{CFG} & \emph{Initial} & \emph{Final} &
	\emph{Common} & \emph{Us} & \emph{ILP} & \emph{Us} & \emph{Min} & \emph{Max} & \emph{MPA} \\
	\hline 
	matmult  & 111 & 130 & 5 & 1105 & 1 & 0  & 0.01 & 0.00 & 3.88 & 0.31 \\
	cnt      & 153 & 284 & 3 & 2278 & 2  & 8 & 0.15 & 0.00 & 3.59 & 30.4 \\
	fft      & 391 & 453 & 8 & 2968 & 4 & 16 & 0.00 & 0.00 & 1.51 & - \\
	compress & 694 & 906 & 3 & 4760 & 11 & 40  & 0.02 & 0.01 & 0.03 & - \\
	lift     & 814 & 1799 & 5 & 5130 & 19 & 40 & 1.51 & 0.05 & 2.29 & - \\
	adpcm    & 2032 & 2211 & 3 & 10688 & 67 & 272  & 0.01 & 0.01 & 0.33 & - \\
	aes\_enc & 2205 & 2651 & 2 & 4914 & 30 & 260 & 0.04 & 0.03 & 0.04 & - \\
	powerwindow & 3738 & 4453 & 24 & 45702 & 224 & 4192 & 0.01  & 0.01 & 1.43 & - \\
	fbw      & 10612 & 27251 & 2 & 36940 & 1198 & 8960 & 2.62 & 0.03  & 7.05 & - \\
	audiobeam & 12299 & 47248 & 37 & 56566 & 1222 & 12824 & 0.12 & 0.00 & 0.49 & - \\
	mpeg2    & 38612 & 1658109 & 3 & 267332 & 12221 & $>$ 1 week & -  & - & - & - \\
	\hline 
\end{tabular}
}
\caption{\label{fig:results}Benchmarking results}
\end{table}

The results of our experiments are shown in
Table~\ref{fig:results}. First, we detail the size of the WCET
formulae computed by our approach. Column \emph{CFG} shows the number
of basic blocks in the CFG. Column \emph{Initial} shows the size (the
number of operands) of the WCET formula before simplification, while
Column \substituted{\emph{Reduced}}{\emph{Final}} shows the formula
size after simplification. In most cases, the size of the
non-simplified formula, which also corresponds to the size of the CFT,
is close to the size of the CFG. Differences are due to the presence
of structure-breaking instructions (such as \emph{goto}, \emph{break},
\emph{continue}, \emph{return} in the middle of a function), which
force basic block aliasing in the CFG to CFT conversion
algorithm. This is especially true for the \emph{mpeg2}, and to a
lesser extent for \emph{lift}, \emph{audiobeam}, and \emph{fbw}
benchmarks. For all benchmarks, the size of the simplified formula is
very small and is related to the number of loops whose iteration count
depends on the parameter.

Then, we compare our approach with an IPET approach. Comparison is
performed according to two criteria: WCET analysis time, and pessimism
of the resulting WCET. The target hardware is an ARM processor with a
set-associative LRU instruction cache (the data cache is not taken
into account). The processor pipeline is analyzed with the exegraph
method~\cite{exegraph} and the instruction cache is modeled using
cache categorization~\cite{ferdinand}. The target instruction cache
used in the analysis \substituted{has $4$ way bits, $8$ row bits, and $4$ block
bits}{has $64$ Kbytes, $16$ ways, and blocks of $16$ bytes}. 
We chose a small cache to highlight the impact of the cache on
the execution time for such small benchmarks. The instruction cache
miss latency was assumed to be $10$ cycles.  Each benchmark is
analyzed as a standalone task, without any modeling of the operating
system. To perform the preliminary steps of the WCET analysis (program
path analysis, CFG building, loop bounds estimation, pipeline and
cache modeling), we rely on OTAWA (version 1.0), an open source WCET
computation tool~\cite{otawa}. These steps are common to the IPET
approach and to our approach. For the remaining steps, in the case of
the IPET approach, we use GNU lp\_solve ILP solver~\cite{lpsolve}. Our
approach was coded in Python, and executed with PyPy 2.4.0. We took
the mean time for 1000 executions of our algorithm, to compensate for
PyPy's slow start speed. To compare the WCET estimates, we instantiate
our WCET formula by assigning to the parameter the constant value used
in the IPET experiment.

The \emph{Common} column represents the time spent by OTAWA for the
preliminary steps (common to IPET and our approach), while the
\emph{Us} (our approach) and \emph{ILP} columns correspond to the time
spent for the remaining steps. The WCET evaluation time is essentially
linear in the size of the CFT in our approach and noticeably lower
than the evaluation time for the IPET approach. Notice that lp\_solve
did not find a solution for mpeg2 after one week of execution
time. Furthermore, let us emphasize that computing the WCET for
different parameter values with the IPET approach requires to run the
whole analysis (\emph{Common}+\emph{ILP}) for each parameter value,
while we only need to do the analysis (\emph{Common}+\emph{Us}) once
and then instantiate the formula for each parameter
value.

WCET pessimism is measured in comparison with the IPET result. The
\emph{Us} column represents the value of the pessimism with our
approach for a fixed value of the parameter (the same value as the one
used for the IPET approach). The \emph{Min} and \emph{Max} columns
represent respectively the minimum pessimism and maximum pessimism (in
percentage) for varying values of the parameter between $1$ and
$1000$.
We observed that, in general, the percentage of pessimism decreases
with the value of the parameter, approximately with an hyperbolic
shape. The pessimism of our approach is much lower than that of the
MPA approach (results extracted from \cite{mpa} are reported in column
MPA). It is also extremely low compared to the IPET
approach. Pessimism in our approach can be attributed to the following
causes: (1) the reduced expressiveness of our CFT annotations (as
opposed to ILP constraints) and (2) paths existing in the CFT but not
in the CFG. \added{Experiments show that the amount of pessimism does
  not depend on the size of the CFG.}


%% file: conclusion.tex
In this paper we presented a novel technique for parametric WCET
analysis, which follows a completely new approach based on symbolic
computation of WCET formulas. Experiments show very promising results:
execution time is lower than the traditional non-parametric IPET
technique and over-approximation of the WCET (compared to IPET) is
extremely low.


We believe that symbolic WCET computation paves the way for many
future works, ranging from purely static analyses, for instance, cache
analysis, to more complex dynamic analyses that will
help building adaptive real-time systems.

One of the main limitations of our method is that it is not
possible to specify constraints relating different parameters, which
may prevents some simplifications in the formulae. Furthermore, some
constraints used in IPET (i.e. some types of unfeasible paths)
cannot be easily represented with context annotations. We plan to
extend context annotations in future works to solve these issues.


%% file: proofs.tex
\section{CFG to CFT}
\label{app:CFG-CFT}

In this appendix, we prove the correctness of our translation from a CFG
to a CFT. Namely, we prove that any valid path in the CFG is also a
valid path in the CFT.

\subsection{Execution paths in a hierarchical DAG}

We have already defined the set of feasible execution paths for a CFG
($\gpaths{G, e}$) and for a CFT ($\tpaths{t}$). We will now define the
function $\dpaths{D}$ that returns the set of feasible paths of a
\emph{hierarchical} DAG $\mathcal{D}$. Since a DAG is a particular case
of graph, $\gpaths$ can also be applied to a DAG, however, an important
difference between both functions is that $\dpaths$ explores recursively
the sub-paths of hierarchical nodes appearing in the DAG.

\begin{mydef}
  Let $\mathcal{D}$ be a DAG. The set of execution paths of $\mathcal{D}$ is defined
  as:
  \[
    \dpaths{\mathcal{D}, e} = \bigcup_{p \in \gpaths{\mathcal{D}, e}} \spaths{p}
  \]
  \noindent where
  \begin{align*}
    \spaths{p.n}&=
    \begin{cases}
      \{ q = q_1 @ q_n | q_1 \in \spaths{p} \wedge
      q_n \in \vpaths{n} \}
      &\text{(if }n\text{ is hierarchical)}\\
      \{ q = q_1 . n   | q_1 \in \spaths{p}\} &\text{(otherwise)}
    \end{cases}\\
    \spaths{\epsilon} &= \{\epsilon\}
  \end{align*}
  \noindent and
  \begin{align*}
    \vpaths{L_h} = \{ &p = p_1 @ ... @ p_{x_h} @ p_e | \forall i, 1\leq i\leq x_h, p_i.h_{next} \in \dpaths{\mathcal{D}_h, h_{next}} \\
    & \wedge p_e.h_{exit} \in \dpaths{\mathcal{D}_h, h_{exit}} \}
  \end{align*}
  \indent where $\mathcal{D}_h, h_{next}$ and $h_{exit}$ are
  respectively the DAG, the next node and the exit node corresponding to
  hierarchical node $L_h$.
\end{mydef}

\subsection{Transformation correctness}

\begin{figure}[htbp]
  \centering
  \begin{tikzpicture}[scale=0.8]
    \tikzset{vertex/.style = {shape=circle,draw,minimum size=1.5em}}
    \tikzset{edge/.style = {->,> = latex'}}
    \tikzset{forced/.style = {shape=circle,draw,minimum size=1.5em,fill=gray!30}}
    \tikzset{pred/.style = {shape=circle,draw,minimum size=1.5em,pattern=north west lines,pattern color=gray!60}}
    \node[forced] (b1) at  (0,0) {$b_1$};
	\node[vertex] (a) [above right of=b1]  {$a$};
	\node[vertex] (b) [above right of=a]  {$b$};
	\node[vertex] (c) [right of=a] {$c$};
	\node[pred] (d) [right of=c] {$d$};
	\node[forced] (b2) [below right of=d]  {$b_2$};
	\node[vertex] (e) [below right of=b1] {$e$};
	\node[vertex] (f) [right of=e] {$f$};
	\node[pred] (g) [right of=f] {$g$};
	\node[pred] (h) [below right of=b2] {$h$};
	\node[pred] (i) [above right of=b2] {$i$};
	\node[forced] (b3) [above right of=h] {$b_3$};
	\node (entry) [left of=b1] {};
	\node (exit) [right of=b3] {};
	\draw[edge] (entry) to (b1);
	\draw[edge] (b3) to (exit);
	\draw[edge] (b1) to (a);
	\draw[edge] (a) to (b);
	\draw[edge] (a) to (c);
	\draw[edge] (b) to (d);
	\draw[edge] (c) to (d);
	\draw[edge] (d) to (b2);
	\draw[edge] (b1) to (e);
	\draw[edge] (e) to (f);
	\draw[edge] (f) to (g);
	\draw[edge] (g) to (b2);
	\draw[edge,bend left=60] (e) to (g);
	\draw[edge] (b2) to (h);
	\draw[edge] (b2) to (i);
	\draw[edge] (h) to (b3);
	\draw[edge] (i) to (b3);
	\node (d1a) [above left=22pt and 2pt of a] {};
	\node (d1b) [below right=2pt and 2pt of d] {};
	\coordinate[shift={(-6pt,6pt)}] (d1l) at (d1a.center);
	\draw [style=dotted] (d1a) rectangle (d1b);
	\node (d2a) [above left=4pt and 2pt of e] {};
	\node (d2b) [below right=2pt and 2pt of g] {};
	\coordinate[shift={(6pt,-6pt)}] (d2l) at (d2b.center);
	\draw [style=dotted] (d2a) rectangle (d2b);
	\node (d3a) [above left=2pt and 2pt of i] {};
	\node (d3b) [below right=2pt and 2pt of i] {};
	\coordinate[shift={(-6pt,6pt)}] (d3l) at (d3a.center);
	\draw [style=dotted] (d3a) rectangle (d3b);
	\node (d4a) [above left=2pt and 2pt of h] {};
	\node (d4b) [below right=2pt and 2pt of h] {};
	\coordinate[shift={(6pt,-6pt)}] (d4l) at (d4b.center);
	\draw [style=dotted] (d4a) rectangle (d4b);
	\draw (d1l) node {$D_{1,1}$};
	\draw (d2l) node {$D_{1,2}$};
	\draw (d3l) node {$D_{2,1}$};
	\draw (d4l) node {$D_{2,2}$};
  \end{tikzpicture}
  \caption{\label{fig:decomposing} Decomposing the DAG}
\end{figure}
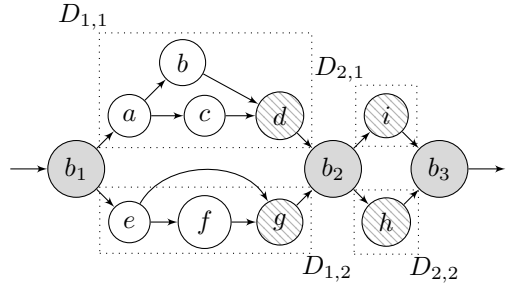

We will proceed in two steps: first we will establish a correspondence
between DAG execution paths and tree execution paths, then
between CFG execution paths and DAG execution paths.

We will now present a graph decomposition technique on which our proof
relies. Let $\mathcal{N}=\{b_1, \ldots,b_n\}$ denote the set of forced
passage nodes of $\mathcal{D}$ towards $\mathcal{D}_e$. Then,
$\mathcal{D}$ can be decomposed into a set of DAGs $\mathcal{D}_{i,j}$,
where $j = 1, \ldots, i_k$ is the $j$-th predecessor of $b_{i+1}$. DAG
$\mathcal{D}_{i,j}$ contains all nodes between $b_i$ (excluded) and the
$j$-th predecessor of $b_{i+1}$ (included), and all related edges. If
$b_i$ is a hierarchical node, we denote $Df_i$ the DAG representing the
corresponding loop (if $b_i$ is a basic block, $Df_i$ is not defined).

Figure~\ref{fig:decomposing} shows such a decomposition. In this
example, the forced passage nodes are shown in gray, and their
predecessors are represented by a striped pattern. The DAG is decomposed
into sub-DAGs $D_{1,1}$, $D_{1,2}$, $D_{2,1}$ and $D_{2,2}$ (plus a
single node DAG for each forced passage node).

\begin{lemma}\label{lm:dpath-in-tpath}
  Let $\mathcal{D}$ be a DAG. Let $t=\Pp(\mathcal{D}, \mathcal{D}_s,
  \mathcal{D}_e)$. We have:

  \[ \dpaths{D,\mathcal{D}_e}\subseteq \tpaths{t}\]
\end{lemma}

\begin{proof}
  The proof is done by induction on the graph decomposition presented
  above. The base of the induction corresponds to the case where
  $\mathcal{D}$ consists only of a chain of forced passage basic
  blocks. Due to the definition of basic blocks though, this chain would always
  consist of a single basic block. Thus proving the induction base is
  trivial.

  Let us now prove the induction step. Let $t_{i,j}=\Pp(\mathcal{D}_{i,j}, \mathcal{D}_{{i,j}_s},
  \mathcal{D}_{{i,j}_e})$, for any appropriate values of $i$ and $j$.
  Let $tfb_i=\Pp(\mathcal{D}f_i, {\mathcal{D}f_i}_s,
  {\mathcal{D}f_i}_n)$, and let $tfe_i=\Pp(\mathcal{D}f_i,
  {\mathcal{D}f_i}_s, {\mathcal{D}f_i}_e)$.

  We must now prove the induction step: assuming Inclusions~\ref{dp1},
  \ref{dp2}, \ref{dp3}, prove Inclusion~\ref{dp4}.

  \begin{align}
    \forall i,j, \dpaths{\mathcal{D}_{i,j},\mathcal{D}_{{i,j}_e}}  &\subseteq \tpaths{t_{i,j}}\label{dp1}\\
    \forall i, \dpaths{\mathcal{D}f_i, {\mathcal{D}f_i}_n} & \subseteq
                                                             \tpaths{tfb_i}\label{dp2}\\
    \forall i, \dpaths{\mathcal{D}f_i, {\mathcal{D}f_i}_e} &\subseteq \tpaths{tfe_i}\label{dp3}\\
    \dpaths{\mathcal{D},\mathcal{D}_e} &\subseteq \tpaths{t}\label{dp4}
  \end{align}
  
  To simplify the notation, we will assume that each time a variable
  named $i$ is introduced in some equation in the proof, it is
  constrained to $1, \ldots, n$. Similarly, when $j$ is introduced, it
  is constrained to $1, \ldots, i_k$.

\medskip

For any path $p$ in $\dpaths{\mathcal{D}, \mathcal{D}_e}$, $p$ can be
expressed as $p=pf_1 @ p_1 @ pf_2 @ ... @ p_{n-1} @ pf_n$, where the
$pf_i$ terms are the path segments corresponding to the execution of
forced passage nodes, and $p_i$ terms are the path segments
corresponding to the execution between these forced passage nodes.

	For all $i$, if $b_i$ is a basic block, then let $tf_i$ =
        $\Leaf(b_i)$. Otherwise, let $tf_i = Loop(tfb_i,tfe_i)$.

	Let us show that $\forall i$, $pf_i \in \tpaths{tf_i}$.  If
        $b_i$ is a basic block, then we have
        $pf_i = \{ b_i \}\in\tpaths{tf_i}$.  If $b_i$ is a hierarchical
        node, then we have $pf_i \in \vpaths{b_i}$. Thanks to induction
        hypothesis,
        $\dpaths{\mathcal{D}f_i, {\mathcal{D}f_i}_n} \subseteq
        \tpaths{tfb_i}$
        and
        $\dpaths{\mathcal{D}f_i, {\mathcal{D}f_i}_e} \subseteq
        \tpaths{tfe_i}$.  Due to the definition of $\vpaths{b_i}$, $pf_i \in \tpaths{tf_i}$.

  
	We have $\forall i, \exists j, p_i \in \dpaths{\mathcal{D}_{i,j}, b_{i+1}}$.
	Thus, thanks to the induction hypothesis, $\forall i, \exists j,
        p_i \in \tpaths{t_{i,j}}$. Thanks to the definition of the function $\tpaths{}$ on the $\Alt$ node, we have $\forall i, p_i \in \tpaths{\Alt(t_{i,1}, ..., t_{i,{i_k}})}$.

  As a consequence, we have $p \in $ $ \tpaths{\Seq(tf_1, $ $ \Alt(t_{1,1}, $ $ ..., $ $ t_{1_k}), ..., $ $ \Alt(t_{{n-1},1}, ..., t_{{n-1},k}), tf_n)}$.

\medskip

  Now, we must prove that this corresponds to the structure of the tree
  built by our algorithm. By examining the algorithm, we see that $t$ is a $\Seq$ node, whose children list alternates between $\Leaf$ nodes representing the forced passage nodes, and $\Alt$ nodes (line \substituted{18}{16}) corresponding to possible paths between forced passage nodes.

  The tree representing the forced passage node $b_i$ is either $\Leaf(b_i)$, if $b_i$ is a basic block (line \substituted{20}{18}), or $\MyLoop(tfb_i, tfe_i)$,
  otherwise (line \substituted{23-26}{21-24}). The definition of this tree is thus that of $tf_i$.
  
  Furthermore, each child tree of one of the $\Alt$ nodes represents the paths between a \emph{forced passage} node, and a predecessor of the next \emph{forced passage} node (the test at line \substituted{6}{4} prevents the double counting of the \emph{forced passage} nodes).
  
  Therefore, we have $t = \Seq(tf_1, $ $ \Alt(t_{1,1}, ..., t_{1_k}),
  ...,$ $  \Alt(t_{{n-1},1}, ..., t_{{n-1},k}), $ $ tf_n)$. As a
  consequence, $p \in \tpaths{t}$ and finally $\dpaths{\mathcal{D}, \mathcal{D}_e} \subseteq \tpaths{t}$.
\end{proof}

Now we can proceed to the final correctness theorem.

\begin{theorem}
  \label{th:path-validity}
  Let $G$ be a CFG and let $G_e$ denote the exit node of $G$. Let
  $\mathcal{D}=DAG(G,\top)$ and let $t=\Pp(\mathcal{D},
  \mathcal{D}_s, \mathcal{D}_e)$. We have:
  \[ \gpaths{G,G_e} \subseteq \tpaths{t}\]
\end{theorem}

\begin{proof}
  Let $D=DAG(G,\top)$. All we need to prove now is that
  $\gpaths{G}\subseteq \dpaths{\mathcal{D},\mathcal{D}_e}$. The
  problem of reducing the CFG into a hierarchy of DAGs is a classical
  problem in compiler theory. Our method is similar to the one
  described in~\cite{hecht1972flow}, so we take its correctness for granted.
\end{proof}

\section{WCET correctness}

\label{app:WCET-correctness}

In this appendix, we show that the WCET obtained with our approach is
greater than the execution time of any feasible path in the CFT. Since
we also proved that any paths of the CFG is also a path of the CFT
obtained by our translation, these two properties ensure that the WCET
computed by our approach is greater than the execution time of any
feasible path in the CFG, which establishes the correctness of our
approach.

Let $\eval(\eta, e, n) \equiv \sum_{i=1}^n{(\eta \otimes e)[i]}$. We
want to prove that the WCET estimation for the program, provided by
function $\eval$, is an upper bound on the execution time of any path in
the tree $t$. The proof strategy is the following:

\begin{itemize}
\item We first define a property of the abstract WCET on a control-flow
  tree. The property is verified only if the abstract WCET is a valid
  representation of the tree's many possible execution times;

\item We then show that our function $\gamma$ provides an abstract
  WCET which verifies the property mentioned above;

\item Finally, we show that this property implies that the WCET estimation
  for the program is an actual upper bound.
\end{itemize}

We start by introducing an helper function \prep\ (for \emph{path
repetition}). It is a generalization of \tpaths{}\ that computes all
the paths in $n$ repetitions of $t$, considering that an external loop $l$
of $t$ has been entered $e$ times:
\begin{mydef}
  \label{def:prep}
  Let $\prep(t, e, n)$ be defined as follows:
  \begin{align*}
    \prep(t, e, n) = \{p|&\exists p_1, \ldots, p_n \in \tpaths{t},p = p_1\cat\ldots\cat p_n , \\
    &\forall (t^{\prime}, l, m) \in \annSet(t), l \notin t \implies \occ(\tpaths{t^{\prime}}, p) \le e \cdot m\}
  \end{align*}
\end{mydef}

If $t$ is the whole program, then $\prep(t, 1, 1) = \tpaths{t}$ (in that
case, there is no loop containing $t$, so all annotations in
$\annSet(t)$ refer to loops inside $t$).

We are now ready to state our predicate.
\begin{mydef}
  \label{def:propV}
  $V(t, \eta)$ is a predicate representing the fact that $\eta$ is
  a valid abstract WCETs for control-flow tree $t$:
  \begin{align*}
    V(t, \eta) \equiv \forall e, n \in \mathbb{N}, p \in \prep(t, e, n), \wcet{p} \le \eval(\eta , e, n) 
  \end{align*}	
\end{mydef}

This property is actually a generalization of
the property we want to prove, i.e. that $\eval(\eta, 1, 1)=\eta[1]$ is a correct upper
bound for any possible execution of a tree $t$.

Then, the following theorem states that the function $\gamma$ computes an abstract WCET that satisfies the property $V$.
\begin{theorem}\label{lm:eval-valid}
  $\forall t \in \mathcal{T}, (l, \eta) = \gamma(t) \implies V(t, \gamma(t))$.  
\end{theorem}

First, we state a property on
$\gamma$ that will be useful later during the proof.
\begin{lemma}\label{lm:most-internal-loop}
  Let $\gamma(t) = (l,\eta)$. Then:
  \[ 
  \forall (t^{\prime}, l^{\prime}, m) \in \annSet(t), l^{\prime} \notin t \implies l \sqsubseteq l^{\prime}.
  \]
\end{lemma}
\begin{proof}
  By definition of $\gamma$ and $\omega$, $l$ is always computed as
  the intersection between external loops. So, it can never happen
  that $l$ refers to a loop that is more external than a loop
  contained within an annotation in $t$.
\end{proof}

We prove the theorem by induction on the structure of the control-flow
tree. We start by proving that, if the property is valid for the result of
$\omega$, then it is also valid for the result of $\gamma$.
\begin{lemma}\label{lm:V-annotation-on-root}
  Let $t$ be a control-flow tree, and let $\ann(t) = (t, l_1, k)$ be
  its annotation. Let $t^\prime$ be the same tree on which the
  annotation on $t$ has been replaced by the empty annotation
  $(t^\prime, \top, \infty)$.  
  Let $\omega(t^\prime) = (l^\prime, \eta^\prime)$ and let $\gamma(t)
  = (l, \eta)$. Then:
  \[
  V(t^\prime, \eta^\prime) \implies V(t, \eta) 
  \]
\end{lemma}
\begin{proof}
  Clearly, $\gamma(t^\prime) = \omega(t^\prime) = \omega(t)$ because
  function $\omega$ does not consider the annotation on the root of
  $t$.

  For all $e, n \in \mathbb{N}$, let $M = \max(e \cdot k, n)$. 

  \begin{enumerate}
  \item by definition, $\prep(t,e,n) = \prep(t^\prime, e, M)$;
  \item by definition, $\eval(\eta,e,n) = \eval(\eta^\prime, e, M)$.
  \end{enumerate}  
  From item 1, it follows that $\forall p \in \prep(t, e, n)$ we have
  also that $p \in \prep(t^\prime, e, M)$. 

  From $V(t^\prime, \eta^\prime)$, it follows that $\wcet{p} \le
  \eval(\eta^\prime, e, M)$. From item 2, $\eval(\eta^\prime, e, M) =
  \eval(\eta,e,n)$ which proves the lemma.
\end{proof}

To prove Theorem~\ref{lm:eval-valid}, we consider each case of the
inductive definition of the CFT separately (\Seq, \Alt, \MyLoop).
\begin{lemma}\label{lm:eval-seq}
  Let $t = \Seq(t_1, t_2)$, and let $(l, \eta) = \gamma(t)$, $(l_1,
  \eta_1) = \gamma(t_1)$, and $(l_2, \eta_2) = \gamma(t_2)$.  
  Then,
  \[
  \forall e,n \in \mathbb{N},\quad V(t_1, \eta_1) \wedge V(t_2, \eta_2) \implies V(t, \eta)
  \]
\end{lemma}
\begin{proof}
  Let $t^\prime$ be the same tree as $t$ but without the annotation on
  the root node and let $(l^\prime, \eta^\prime) = \omega(t^\prime)$. 

  By definition of function $\eval$, we have:
  \begin{align*}
    \eval(\eta^\prime, e, n) &= \sum_{i=0}^{n-1} (\eta^\prime \otimes e)[i] =
     \sum_{i=0}^{n-1} ((\eta_1 \oplus \eta_2) \otimes e)[i] = 
     \sum_{i=0}^{n-1} (\eta_1 \otimes e)[i] + \sum_{i=0}^{n-1} (\eta_2 \otimes e)[i] \\
    &= \eval(\eta_1, e, n) + \eval(\eta_2, e, n).
  \end{align*}
  By definition of predicate $V$:
  \begin{align*}
    \forall p_1 \in \prep(t_1, e, n), \; &\wcet{p_1} \leq \eval(t_1, e, n) \\ 
    \forall p_2 \in \prep(t_2, e, n), \; &\wcet{p_2} \leq \eval(t_2, e, n)
  \end{align*}
  Any path $p \in \prep(t, e, n)$ is a permutation of some $p_1 \cat p_2$, hence 
  \[
    \wcet{p} \leq \eval(\eta_1, e, n) + \eval(\eta_2, e, n) = \eval(\eta^\prime, e, n)
  \]
  and this proves that $V(t^\prime, \eta^\prime)$ holds. From Lemma
  \ref{lm:V-annotation-on-root}, it follows that $V(t, \eta)$ also holds.
\end{proof}

\begin{lemma}\label{lm:eval-alt}
  Let $t = \Alt(t_1, t_2)$, $(l, \eta) = \gamma(t)$, and $(l_1, \eta_1) = \gamma(t_1)$, and $(l_2, \eta_2) = \gamma(t_2)$.
  Then, 
  \[
  \forall e,n \in \mathbb{N},\quad V(t_1, \eta_1) \wedge V(t_2, \eta_2) \implies V(t, \eta)
  \]
\end{lemma}
\begin{proof}
  Let $t^\prime$ be the same tree as $t$ but without the annotation on
  the root node and let $(l^\prime, \eta^\prime) = \omega(t^\prime)$. 
  By definition of functions $\omega$ and $\gamma$, $\eta^\prime = \eta_1 \uplus \eta_2$. It follows that
  \[
  \eta^\prime \otimes e = (\eta_1 \uplus \eta_2) \otimes e = (\eta_1 \otimes e) \uplus (\eta_2 \otimes e).
  \]
	
  From the $n$ greatest elements of $\eta^\prime \otimes e$, we have $x$
  elements coming from $\eta_1 \otimes e$, and $y$ elements coming
  from $\eta_2 \otimes e$. We note that we can have several valid
  values of $x$ and $y$ if there are shared time values between $\eta_1$ and
  $\eta_2$.

  By definition, we have:
  \begin{align*}
    \eval(\eta^\prime, e, n) &=\sum_{i=0}^{n-1} (\eta^\prime \otimes e)[i] =
                     \sum_{i=0}^{n-1} ((\eta_1 \otimes e) \uplus (\eta_2 \otimes e))[i] \geq
                     \sum_{i=0}^{x-1} (\eta_1 \otimes e)[i] + \sum_{i=0}^{y-1} (\eta_2 \otimes e)[i]
  \end{align*}
  The last inequality is true for any choice of $x$ and $y$ such that
  $x+y=n$, because we pick the $x$ greatest elements from $\eta_1
  \otimes e$ and the $y$ greatest elements from $\eta_2 \otimes
  e$. Moreover, the sum of the $n$ greatest elements of $\eta^\prime \otimes
  e$ is never inferior to the sum of the $x$ greatest elements of
  $\eta_1 \otimes e$ and the $y$ greatest elements of $\eta_2 \otimes
  e$.
  
  Now, let $p_{max}$ be the worst-case path of $\prep(t^\prime, e,
  n)$. Because of the definition of $\prep$ on \Alt\ nodes, we can
  find $x$ and $y$ such that $p_1 \in \prep(t_1, e, x)$ and $p_2 \in
  \prep(t_2, e, y)$, and such that $p_{max}$ is a permutation of nodes
  from $p_1$ and $p_2$.  Obviously, we have $\wcet{p_{max}} = \wcet{p_1} +
  \wcet{p_2}$.
  Because of the induction hypothesis, we have $\wcet{p_1} \le
  \eval(\eta_1, e, x)$ and $\wcet{p_2} \le \eval(\eta_2, e,
  y)$. Therefore $\wcet{p_{max}} \le \eval(\eta^\prime, e, n)$, and
  this proves that $V(t^\prime, \eta^\prime)$ holds.  From Lemma
  \ref{lm:V-annotation-on-root}, it follows that $V(t, \eta)$ also
  holds.
\end{proof}

\begin{lemma}\label{lm:eval-loop}
  Let $t = \MyLoop(h, t_b, x_h, t_e)$, $(l, \eta) = \gamma(t)$, $(l_b,
  \eta_b) = \gamma(t_b)$, $(l_e, \eta_e) = \gamma(t_e)$.  Then:
  \[
  \forall e,n \in \mathbb{N},\quad V(t_b, \eta_b) \wedge V(t_e, \eta_e) \implies V(t, \eta)
  \]
\end{lemma}
\begin{proof}
  Let $t^\prime$ be the same tree as $t$ but without the annotation on
  the root node and let $(l^\prime, \eta^\prime) = \omega(t^\prime)$.   

  If $l_b = l_h$, from the definition of $\omega$, it follows that the
  estimated time for one full execution of loop $l$ is constant. Let
  us name this constant $c = \eval(\eta_b, 1, x_h) =
  \sum_{i=0}^{x_h-1} \eta_b[i]$. By definition of $\gamma$ and
  $\omega$:
  \[
  \eval(\eta^\prime, e, n) = c n + \eval(\eta_e, e, n).
  \]

  For all $p \in \prep(t^\prime, e, n)$, $\exists p_1, \ldots, p_n \in
  \prep(t_b, 1, x_h)$ and $\exists p_e \in \prep(t_e, e, n)$, such
  that $p$ is a permutation of $p_1 \cat \ldots\cat p_n\cat p_e$.  We
  have $\wcet{p} = \wcet{p_1} + \ldots + \wcet{p_n} +
  \wcet{p_e}$. Also, $\forall k, \wcet{p_k} \le \eval(\eta_b, 1, x_h)$.
  Therefore,
  \[
  \wcet{p} \le n \cdot \eval(\eta_b, 1, x_h) + \eval(\eta_e, e, n) = \eval(\eta^\prime, e, n).
  \]

  Notice that we can rule out case $l_e = l_h$, by definition of context annotations.

  If $l_b \ne l_h$, then by definition of $\gamma$ and $\omega$
  functions, we have 
  \[
   \eta^\prime[i] = \sum_{j=i \cdot x_h}^{i \cdot x_h + x_h - 1}\eta_1[j].
  \]  
  We know that $n$ is a multiple of $e$. Let $n = k \cdot e$. We have:
  \begin{align*}
    \eval(\eta^\prime, e, n) &= \eval(\eta^\prime, 1, k) \cdot e  \\
      &= e \cdot \sum_{i=0}^{k -1} \sum_{j=i \cdot x_h}^{i \cdot x_h + x_h - 1}(\eta_b) [j] + \eval(\eta_e, e, n) = 
       e \cdot \sum_{i=0}^{k \cdot x_h - 1} \eta_b[i] + \eval(\eta_e, e, n)  \\
      &= e \cdot \eval(\eta_b, 1, k \cdot x_h) + \eval(\eta_e, e, n) = 
       \eval(\eta_b, e, n \cdot x_h) + \eval(\eta_e, e, n).
  \end{align*}

  Since $l_b \ne l_h$, and from Lemma \ref{lm:most-internal-loop}, we
  know that no annotation in $t$ refers to the current loop. 
  Therefore, for all $p \in \prep(t^\prime, e, n)$, $p$ can be expressed as
  the permutation of $p_b \cat p_e$, where paths $p_b \in \prep(t_b,
  e, n \cdot x_h)$ and $p_e \in \prep(t_e, e, n)$.
  Then:
  \[
   \wcet{p} = \wcet{p_b} + \wcet{p_e} \le \eval(\eta_b, e, n \cdot x_h) + \eval(\eta_e, e, n) = \eval(\eta^\prime, e, n).
  \]
  This proves that $V(t^\prime, \eta^\prime)$ holds.  From Lemma
  \ref{lm:V-annotation-on-root}, it follows that $V(t, \eta)$ also
  holds.
\end{proof}

We can now conclude on
the validity of our complete WCET evaluation method.

\begin{theorem}
  Let $G$ a CFG. Let $\mathcal{D}=DAG(G,\top)$ and let
  $t=\Pp(\mathcal{D}, \mathcal{D}_s, \mathcal{D}_e)$. Let
  $(l,\eta)=\gamma(t)$. We have: $\forall p\in\gpaths{G,G_e}, \wcet{p} \leq \eval(\eta,1,1)$
\end{theorem}

\begin{proof}
  Consequence of Theorem~\ref{th:path-validity} and Theorem~\ref{lm:eval-valid}.
\end{proof}
